\documentclass[1p]{elsarticle}
\usepackage{lineno}
\usepackage{hyperref}
\modulolinenumbers[5]
\usepackage{fix-cm}
\usepackage{mathtools}
\usepackage{amsmath,amsthm}
\usepackage{graphicx}
\usepackage{latexsym,enumerate,amssymb,amsbsy}
\usepackage{graphics,color}
\usepackage{verbatim}
\usepackage{url}
\usepackage{dblfloatfix}
\usepackage{mathptmx}
\usepackage{pgf,tikz}
\usetikzlibrary{trees,automata,arrows,shapes.geometric}
\usepackage{xfrac}
\usepackage{amsfonts}
\usepackage{stmaryrd}
\usepackage{algorithm}
\usepackage[noend]{algpseudocode}
\usepackage{url}
\usepackage{colortbl}

\newtheorem{proposition}{Proposition}
\newtheorem{theorem}[proposition]{Theorem}
\newcommand{\proofofref}{}
\newproof{zproofof}{Proof of \proofofref}
\newenvironment{proofof}[1]
{\renewcommand{\proofofref}{#1}\zproofof}
{\endzproofof}
\newtheorem{corollary}[proposition]{Corollary}
\newtheorem{lemma}[proposition]{Lemma}

\newtheorem{example}{Example}

\newtheorem{remark}{Remark}
\newcommand*{\QEDB}{\hfill\ensuremath{\square}}

\newcommand{\etal}{\emph{et al.}}
\newcommand{\eg}{\emph{e.g.}}
\newcommand{\ie}{\emph{i.e.}}

\newcommand{\bbra}[1]{\left\llbracket{#1}\right\rrbracket}
\def\F{{\mathbb{F}}}

\def\bv{{\mathbf{v}}}
\def\bw{{\mathbf{w}}}
\def\bS{{\mathbf{S}}}
\def\bu{{\mathbf{u}}}
\def\bs{{\mathbf{s}}}
\def\bm{{\mathbf{m}}}
\def\0{{\mathbf{0}}}
\def\1{{\mathbf{1}}}
\def\ba{{\mathbf{a}}}
\def\bb{{\mathbf{b}}}
\def\bc{{\mathbf{c}}}

\def\cG{{\mathcal{G}}}

\def\cC{{\mathcal{C}}}
\def\cR{{\mathcal{R}}}

\DeclarePairedDelimiter\abs{\lvert}{\rvert}

\algdef{SE}[DOWHILE]{Do}{doWhile}{\algorithmicdo}[1]{\algorithmicwhile\ #1}

\begin{document}
\begin{frontmatter}
\title{On Greedy Algorithms for Binary de Bruijn Sequences}
		
\author[zuling]{Zuling Chang}
\ead{zuling\_chang@zzu.edu.cn}
		
\cortext[cor1]{Corresponding author}
\author[ntu]{Martianus Frederic Ezerman\corref{cor1}}
\ead{fredezerman@ntu.edu.sg}
		
\author[ntu]{Adamas Aqsa Fahreza}
\ead{adamas@ntu.edu.sg}
		
\address[zuling]{School of Mathematics and Statistics, Zhengzhou University, 450001 Zhengzhou, China.}
		
\address[ntu]{School of Physical and Mathematical Sciences, Nanyang Technological University,\\
21 Nanyang Link, Singapore 637371.}
		
\begin{abstract}
We propose a general greedy algorithm for binary de Bruijn sequences, called {\tt Generalized Prefer-Opposite} (GPO) Algorithm, and its modifications. By identifying specific feedback functions and initial states, we demonstrate that most previously-known greedy algorithms that generate binary de Bruijn sequences are particular cases of our new algorithm.
\end{abstract}
		
\begin{keyword}
Binary periodic sequence \sep de Bruijn sequence \sep greedy algorithm \sep feedback function \sep state graph.
\end{keyword}
		
\end{frontmatter}
	

\section{Introduction}\label{sec:intro}

There are $2^{2^{n-1}-n}$ binary {\it de Bruijn sequence} of order $n$~\cite{Bruijn46}. Each has period $N=2^n$ in which every binary $n$-tuple occurs exactly once. 

Relatively fast generation of a few de Bruijn sequences has a long history in the literature. Any primitive polynomial of degree $n$ generates a maximal length sequence (also known as $m$-sequence) that can be easily modified to a de Bruijn sequence~\cite{Golomb}. There are at least three different algorithms to generate a Ford sequence, which is the lexicograpically smallest de Bruijn sequence. They were treated, respectively, by Fredricksen and Maiorana in~\cite{FM78}, by Fredricksen in~\cite{Fred82}, and by Ralston in~\cite{Ral82}. There are frameworks for construction via {\it successor rules}. More recent, necklace-based, successor rules and their co-necklace variants can be found in several works, \eg, \cite{SWW16}, \cite{Sawada16}, \cite{SWW17}, and \cite{Dra18}.

Many, often less efficient, methods that produce much larger numbers of de Bruijn sequences are known. The majority of them work by joining smaller cycles into de Bruijn sequences. Prominent examples include methods to join cycles generated by pure cycling register or pure summing registers given by Fredricksen in~\cite{Fred75}, Etzion and Lempel in~\cite{EL84}, and Huang in~\cite{Huang90}. Examples of recent works on more general component cycles are the work of Li \etal~\cite{Li16} and Chang \etal~\cite{Chang2019}. 

Here we focus on greedy algorithms to generate de Bruijn sequences. There have also been many known ones, \eg, {\tt Prefer-One} and {\tt Prefer-Same} in~\cite{Fred82} and {\tt Prefer-Opposite} in~\cite{Alh10}. They have since been generalized using the notion of {\it preference functions} in~\cite{Alh12}. A paper of Wang \etal~\cite{WWZ18}, which was presented at SETA 2018, discusses a greedy algorithm based on the feedback function $x_{n-2}+x_{n-1}$ for $n \geq 3$. 

In general one can come up with numerous feedback functions to generate de Bruijn sequences using greedy algorithms. It had unfortunately been rather challenging to confirm which ones of these functions actually work. We show how to circumvent this by specifying feedback functions that come with a \emph{certificate of correctness}.

We state the {\tt Generalized Prefer-Opposite} (GPO) Algorithm and prove sufficient conditions on the feedback functions to ensure that the algorithm indeed generates de Bruijn sequences. This leads us to numerous classes of special feedback functions that can be used to generate de Bruijn sequences via the algorithm. As a corollary, we show that {\tt Prefer-One}, {\tt Prefer-Zero}, as well as others based on preference functions are special cases of the GPO algorithm.

To include even more classes of feedback functions, we put forward suitable modifications of the GPO Algorithm. Several new families of de Bruijn sequences can then be generated in a greedy manner.

After this introduction comes preliminary notions and results in Section~\ref{sec:prelim}. We introduce the GPO Algorithm in Section~\ref{sec:GPO}. Sufficient conditions for the algorithm to produce de Bruijn sequences are then proved. The three subsections describe three respective families of de Bruijn sequences, with some analysis on their properties. Section~\ref{sec:modify} shows how to modify the GPO Algorithm when the sufficient condition is not met. The modification results in numerous instances of successful construction of de Bruijn sequences. Three more families of such sequences are then showcased. We end with a conclusion that, for any de Bruijn sequence $\bS$ of order $n > 2$ and any $n$-string initial state $\bb$, one can always find a feedback function $f$ that the GPO Algorithm can take as an input to produce $\bS$.

\section{Preliminaries}\label{sec:prelim}

Let $0 < k < \ell $ be integers. We denote $\{0,1,2,\ldots,\ell\}$ by $\bbra{\ell}$ and $\{k,k+1,\ldots,\ell\}$ by $\bbra{k,\ell}$. An {\it $n$-stage shift register} is a clock-regulated circuit with $n$ consecutive storage units. Each of the units holds a bit and, as the clock pulses, the bit is shifted to the next stage in line. The output is a new bit $s_n$ based on the $n$ bits $\bs_0= s_0,\ldots,s_{n-1}$ called the
{\it initial state}. The corresponding {\it feedback function} $f(x_0,\ldots,x_{n-1})$ is the Boolean function that, on input $\bs_0$, outputs $s_n$.

With a function $f(x_0,x_1,\ldots,x_{n-1})$, a feedback shift register (FSR) outputs a sequence $\bs=s_0,s_1,\ldots,s_n,\ldots$ satisfying $s_{n+\ell} = f(s_{\ell},s_{\ell+1},\ldots,s_{\ell+n-1})$ for
$\ell = 0,1,2,\ldots$. Let $N$ be the smallest positive integer satisfying $s_{i+N}=s_i$ for all $i \geq 0$. Then $\bs$ is {\it $N$-periodic}
or {\it with period $N$} and one writes $\bs= (s_0,s_1,s_2,\ldots,s_{N-1})$. We call $\bs_i= s_i,s_{i+1},\ldots,s_{i+n-1}$ {\it the $i$-th state} of $\bs$ and states $\bs_{i-1}$ and $\bs_{i+1}$ the {\it predecessor} and {\it successor} of $\bs_i$, respectively.

An $n$-string $c_0,c_1,\ldots,c_{n-1}$ is often written concisely as $c_0c_1 \ldots c_{n-1}$, especially in algorithms and tables. The {\it complement} $\overline{c}$ of $c \in \F_2$ is $1+c$. The complement of a string is the corresponding string produced by taking the complement of each element. The string of zeroes of length $\ell$ is denoted by $\0^{\ell}$. Analogously, $\1^{\ell}$ denotes the string of ones of length $\ell$. Given a state $\bc= c_0,c_1,\ldots,c_{n-1}$, we say that $\overline{\bc}:= \overline{c_0},c_1,\ldots,c_{n-1}$ and $\widehat{\bc}:= c_0,c_1,\ldots,\overline{c_{n-1}}$ are, respectively, the {\it conjugate state} and the {\it companion state} of $\bc$.

A function $f(x_0,x_1,\ldots,x_{n-1})=x_0+g(x_1,\ldots,x_{n-1})$, where $g$ is a Boolean function, is said to be {\it non-singular}. An FSR with non-singular feedback function will generate periodic sequences \cite[p.~116]{Golomb}. Otherwise $f$ is said to be {\it singular} and some states would have two distinct preimages under $f$.

The {\it state graph} of the FSR with feedback function $f(x_0,x_1,\ldots,x_{n-1})$ is a directed graph $\cG_f$ whose vertices are all of the $n$-stage states. There is a directed edge from a state $\bu:=u_0,u_1,\ldots,u_{n-1}$ to a state $\bv:=u_1,u_2,\ldots,f(u_0,u_1,\ldots,u_{n-1})$. We call $\bu$ a {\it child} of $\bv$ and $\bv$ the {\it parent} of $\bu$. We allow $\bu = \bv$, in which case $\cG_f$ contains a loop. A {\it leaf} in $\cG_f$ is a vertex with no child, \ie, a leaf has outdegree $1$ and indegree $0$. We say that a vertex $\bw$ is a {\it descendant} of $\bu$ is there is a directed path that starts at $\bw$ and ends at $\bu$, which in turn is called an {\it ancestor} of $\bw$. A {\it rooted tree} $T_{f,\bb}$ in $\cG_f$ is the largest tree in $\cG_f$ in which one vertex $\bb$ has been designated the {\it root} with the edge that emanates out of $\bb$ removed from $\cG_f$. In this work the orientation is \underline{ towards} the root $\bb$, \ie, $T_{f,\bb}$ is an {\it in-tree} or an {\it arborescence converging to a root} as defined in~\cite[Chapter~6]{Tutte}.

\section{Generalized Prefer-Opposite Algorithm}\label{sec:GPO}

This section describes the GPO algorithm and proves sufficient conditions on the feedback function to guarantee that the algorithm generates de Bruijn sequences. 

For a given feedback function $f(x_0,x_1,\ldots,x_{n-1})$ and initial state $\bb=b_0,b_1,\ldots,b_{n-1}$, the GPO algorithm is presented here as Algorithm~\ref{algo:po}. Notice that it does not always produce de Bruijn sequences. The following Theorem~\ref{thm:po} provides sufficient conditions on the feedback function and the initial state for the GPO algorithm to generate de Bruijn sequence(s) of order $n$.

\begin{algorithm}[ht!]
\caption{{\tt Generalized Prefer-Opposite}}
\label{algo:po}
\begin{algorithmic}[1]
	\renewcommand{\algorithmicrequire}{\textbf{Input:}}
	\renewcommand{\algorithmicensure}{\textbf{Output:}}
	\Require A feedback function $f(x_0,x_1,\ldots,x_{n-1})$ and an initial state $\bb$.
	\Ensure A binary sequence.
	\State{$\bc=c_0, c_1, \ldots ,c_{n-1} \gets \bb$}	
	\Do
	\State{Print($c_0$)}
	\State{$y \gets f(c_0,c_1,\ldots,c_{n-1})$}
	\If{$c_1, c_2, \ldots, c_{n-1}, \overline{y}$ has not appeared before}
		\State{$\bc \gets c_1, c_2, \ldots, c_{n-1},  \overline{y} $} \label{Algpo:line5}
		\Else
		\State{$\bc \gets c_1, c_2, \ldots, c_{n-1}, y$} \label{Algpo:line7}
	\EndIf
	\doWhile{$\bc \neq \bb$}\label{Algpo:line9}
\end{algorithmic}
\end{algorithm}

\begin{theorem}\label{thm:po}
The GPO Algorithm, on input a function $f(x_0,x_1,\ldots,x_{n-1})$ and an initial state $\bb=b_0,b_1,\ldots,b_{n-1}$, generates a binary de Bruijn sequence of order $n$ if the state graph $\cG_f$ of the FSR satisfies the following two conditions.
\begin{enumerate}
	\item All of the states, except for the leaves, have exactly two children.
    \item There is a unique directed path from any state $\bv$ to $\bb$.
\end{enumerate}
\end{theorem}

Before proving the theorem, we establish an important lemma that will be frequently invoked.
\begin{lemma}\label{lemma1}
Let the state graph $\cG_f$ of the FSR with feedback function $f$ satisfy the first condition in the statement of Theorem~\ref{thm:po}. Let the state $\bs$ be a vertex with two children. Let the GPO Algorithm starts with an initial state $\bb \neq \bs$, with $\bb$ not a leaf. By the time the algorithm visits $\bs$ it must have visited both children of $\bs$.
\end{lemma}

\begin{proof}
Suppose that $\bb$ is a leaf. To visit $\bb$ for the second time, one of its possible predecessors, say $\bu=u_0,u_1,\ldots,u_{n-1}$, must have been visited before. At this point, $\bb$ must be of the form $u_1,\ldots,u_{n-1},\overline{f(u_0,u_1,\ldots,u_{n-1})}$. By the rule of the algorithm, since $\bb$ had appeared at the very beginning, the successor of $\bu$ must then be $u_1,\ldots,u_{n-1},f(u_0,u_1,\ldots,u_{n-1})$. This rules $\bb$ out from being visited again. The algorithm will not terminate since Line~\ref{Algpo:line9} in Algorithm~\ref{algo:po} is always satisfied and the iteration continues. Thus, to ensure that the output is a periodic sequence, $\bb$ must not be a leaf.

Consider the state $\bs$ and let $\bv$ and $\bw$ be the two children of $\bs$ in $\cG_f$. Since $\bs$ is not a leaf, one of the children, say $\bw$, must have been its predecessor in the sequence $\bS$ produced by the GPO Algorithm thus far. Suppose, for a contradiction, that the algorithm visits $\bs$ without having visited $\bv$. Consider the other possible successor of both $\bv$ and $\bw$, which is the companion state $\widehat{\bs}$ of $\bs$. Since $\bv$ has not been visited, then $\widehat{\bs}$ must have not been visited as well. The fact that $\bs$ is the actual successor of $\bw$ contradicts the rule of the algorithm since the successor of $\bw$ ought to have been $\widehat{\bs}$.
\end{proof}

\begin{proofof}{Theorem~\ref{thm:po}}
We first show that it is impossible for Algorithm \ref{algo:po} to visit any state twice \emph{before} it reaches the initial state $\bb$ the second time. For a contradiction, suppose that $\bu \neq \bb$ is the first state to be visited twice. The assignment rule precludes visiting the state of the form 
\[
c_1, c_2, \ldots, c_{n-1}, \overline{f(c_0,c_1,\ldots,c_{n-1})}
\]
twice. Hence, $\bu$ must be of the form $c_1, \ldots, c_{n-1}, f(c_0,c_1,\ldots,c_{n-1})$, implying that $\bu$ has as a child the state $c_0, c_1, \ldots, c_{n-1}$ in $\cG_f$ and $\bu$ is not a leaf. By Lemma \ref{lemma1}, when $\bu$ is first visited by the algorithm, its two children must have been visited. The second time $\bu$ is visited, one of its two children must have also been visited twice. This contradicts the assumption that $\bu$ is the first vertex to have been visited twice.

The algorithm continues until it visits $\bb$ again. Any other state, say $\bv \neq \bu$, has two possible successors in $\bS$. At least one of the two must have not been visited yet since their two possible predecessors in $\bS$ include $\bv$.

Now it suffices to show that Algorithm \ref{algo:po} visits all states. Since there is unique directed path to $\bb$ from any other state, all of the other states are descendants of $\bb$. By Lemma \ref{lemma1}, by the time the algorithm revisits $\bb$, it must have visited $\bb$'s two children, even if one of the two is $\bb$ itself. The same lemma implies that the respective child(ren) of these two children of $\bb$ must have been visited beforehand. By repeated applications of the lemma we confirm that all of the descendants of $\bb$ must have been covered in the running of Algorithm \ref{algo:po}.
\end{proofof}

The initial state in {\tt Prefer-One} is $\0^n$. The next bit of the sequence is $1$ if the newly formed $n$-stage state has not previously appeared in the sequence. If otherwise, then the next bit is $0$. For $n=4$, for example, the algorithm outputs $\bS:=(0000~1111~0110~0101)$. {\tt Prefer-Zero}, introduced by Martin in~\cite{Martin1934}, works in a similar manner, interchanging $0$ and $1$ in the rule of {\tt Prefer-One}. Knuth nicknamed {\tt Prefer-Zero} the {\it granddaddy} construction of de Bruijn sequences~\cite{Knuth4A}. For $n=4$, the granddaddy outputs $(1111~0000~1001~1010)$, which is the complement of $\bS$.

\begin{corollary}\label{cor:oz}
Let $n \geq 2$. {\tt Prefer-One} is the special case with $f(x_0,x_1\ldots,x_{n-1})=0$ and $\bb=\0^n$. 
{\tt Prefer-Zero} is the special case with $f(x_0,x_1\ldots,x_{n-1})=1$ and $\bb= \1^n$.
\end{corollary}

To generate de Bruijn sequences by Algorithm~\ref{algo:po} it suffices to find families of feedback functions and initial states that satisfy the conditions in Theorem \ref{thm:po}. There are numerous such combinations with only three families of them explicitly discussed in the present work.

\subsection{Family One}

Our next result gives a family of de Bruijn sequences that can be greedily constructed from de Bruijn sequences of lower orders. It provides an alternative proof to \cite[Theorem 2.4]{Alh12} in the binary case.

\begin{theorem}\label{thm:db}
Let $h(x_0,x_1,\ldots,x_{m-1})$ be a feedback function whose FSR generates a de Bruijn sequence $\bS_m$ of order $m$. We fix a positive integer $n > m \geq 2$. The GPO Algorithm generates the family $\mathcal{F}_1(n; h,m)$ of de Bruijn sequences of order $n$ on input any string of length $n$ in $\bS_m$ as the initial state $\bb$ and
\begin{equation}\label{equ-db}
f(x_0,x_1,\ldots,x_{n-1}) := h(x_{n-m},x_{n-m+1},\ldots,x_{n-1})
\end{equation}
\end{theorem}

\begin{proof}
It suffices to prove that $\cG_f$ with $f$ in Equation (\ref{equ-db}) satisfies the two conditions in Theorem \ref{thm:po}. Since the coefficient of $x_0$ in $f$ is $0$, any non-leaf state has two children. It is immediate to confirm that $\cG_f$ contains a directed cycle whose vertices are the $n$-stage states of $\bS_m$. All other states are vertices in the trees whose respective roots are the states of the above cycle. Since there is a unique directed path from any state to some state $\bu$ in the cycle and there is a unique directed path from $\bu$ to $\bb$, the second condition is satisfied.
\end{proof}

\begin{example}
$\bS_4:=(0000~1101~0010~1111)$ is de Bruijn with
\begin{equation}\label{eq:h}
h(x_0,x_1,x_2,x_3) := 1+ x_0 + x_2 \cdot x_3 + x_1 \cdot x_2 \cdot x_3.
\end{equation}
Table~\ref{table:EDB} provides the $16$ distinct de Bruijn sequences in the $\mathcal{F}_1(6; h, 4)$ family. \QEDB

\begin{table}[h!]
\caption{All $16$ de Bruijn Sequences in $\mathcal{F}_1(6;h,4)$ with $h(x_0,x_1,x_2,x_3)$ in Equation~\ref{eq:h}.}
\label{table:EDB}
\renewcommand{\arraystretch}{1.1}
\centering
\small
\begin{tabular}{c c}
\hline
$\bb$ & The Resulting de Bruijn Sequence \\
\hline
$000011$ & $(00000010 ~ 10010111 ~ 10000111 ~ 0110~0111 ~ 1110 1010 ~ 1101 1100 ~ 1000 1001 ~ 1000 1101)$ \\

$000110$ & $(0000 0010 ~ 1001 0111 ~ 1000 0110 ~ 0111 0110 ~ 0010 0011 ~ 1111 0101 ~ 0110 1110 ~ 0100 1101)$ \\

$001101$ & $(00000010 ~ 10010111 ~ 1000 0110 ~ 1100 1110 ~ 1110 0100 ~ 0100 1100 ~ 0111 1110 1010 ~ 1101)$ \\

$011010$ & $(00000010 ~ 1001 0111 ~ 1000 0110 ~ 1010 1100 ~ 1110 1101 ~ 1100 1000 ~ 1001 1000 ~ 1111 1101)$ \\

$110100$ & $(00000010 ~ 1010 0111 ~ 0110 0110 ~ 0011 1111 ~ 0101 1011 ~ 1001 0010 ~ 1111 0000 ~ 1101 0001)$ \\

$101001$ & $(00000010 ~ 11110000 ~ 11010011 ~ 10110011 ~ 00010001 ~ 11111010 ~ 10110111 ~ 00100101)$ \\

$010010$ & $(00000010 ~ 11110000 ~ 11010010 ~ 00100111 ~ 01100110 ~ 00111111 ~ 01010110 ~ 11100101)$ \\

$100101$ & $(00000010 ~ 11001110 ~ 11011100 ~ 10011000 ~ 11111101 ~ 01111000 ~ 01101001 ~ 01010001)$ \\

$001011$ & $(00000010 ~ 10000110 ~ 10010110 ~ 01110110 ~ 11100100 ~ 01001100 ~ 01111110 ~ 10101111)$ \\

$010111$ & $(00000010 ~ 10000110 ~ 10010111 ~ 01100111 ~ 00100010 ~ 01100011 ~ 11110101 ~ 01101111)$ \\

$101111$ & $(00000010 ~ 10000110 ~ 10010111 ~ 11101100 ~ 11101010 ~ 11011100 ~ 10001001 ~ 10001111)$ \\

$011110$ & $(00000010 ~ 10000110 ~ 10010111 ~ 10110011 ~ 10101011 ~ 01110010 ~ 00100110 ~ 00111111)$ \\

$111100$ & $(00000010 ~ 10000110 ~ 10010111 ~ 10011101 ~ 10010001 ~ 00110001 ~ 11111010 ~ 10110111)$ \\

$111000$ & $(00000010 ~ 10000110 ~ 10010111 ~ 10001000 ~ 11101100 ~ 11111101 ~ 01011011 ~ 10010011)$ \\

$110000$ & $(00000010 ~ 00101010 ~ 00011101 ~ 10011111 ~ 10101101 ~ 11001001 ~ 10001101 ~ 00101111)$ \\

$100001$ & $(00000011 ~ 10110011 ~ 11110101 ~ 10111001 ~ 00110001 ~ 10100101 ~ 11100001 ~ 00010101)$ \\
\hline
\end{tabular}
\end{table}
\end{example}

When $n>m+1$, distinct initial states yield distinct de Bruijn sequences of order $n$, yielding $2^m$ distinct de Bruijn sequences for a valid $f$. But if $n=m+1$, then there are collisions in the output. Two pairs of initial states, namely $(\1^m 0,0 \1^m)$ and $(\0^m 1,1 \0^m)$, generate the same sequence. We formalize the observation in Theorem~\ref{thm:distinct} and supply its proof in the appendix. 

\begin{theorem}\label{thm:distinct}
The number of distinct de Bruijn sequences in $\mathcal{F}_1(n;h,m)$ with $n > m \geq 2$ is
\begin{equation}\label{eq:sizeF1}
\abs{\mathcal{F}_1(n;h,m)} =
\begin{cases}
2^m-2 & \mbox{ if }  n = m + 1, \\
2^m  & \mbox{ if }  n > m+1.
\end{cases}
\end{equation}
\end{theorem}

\subsection{Family Two}
Here is another family of de Bruijn sequences that can be greedily constructed.
\begin{theorem}\label{thm:prod}
Given $n$ and an integer $0<t<n$, let the feedback function be
\begin{equation}\label{equ-prod}
f(x_0,x_1,\ldots,x_{n-1}) :=
1+\prod_{i=t}^{n-1}x_i
\end{equation}
and the initial state $\bb$ be any $n$-stage state of the sequence
$(0,\1^{n-t})$. Then the GPO Algorithm generates the $\mathcal{F}_2(n)$ family of de Bruijn sequences of order $n$.
\end{theorem}

\begin{proof}
First, since the coefficient of $x_0$ in $f$ as given in Equation (\ref{equ-prod}) is $0$, any non-leaf vertex has two children in $\cG_f$. Second, we confirm that $\cG_f$ contains a cycle whose vertices are all of the $n$-stage states of $(0,\1^{n-t})$. All other states are vertices in the trees whose respective roots are the states of this cycle. Checking that the two conditions in Theorem~\ref{thm:po} are satisfied is in the end rather straightforward.
\end{proof}

\begin{remark} We make two notes on the $\mathcal{F}_2(n)$ family. First, the produced sequence is the same as the output of {\tt Prefer-Same} when $t=n-1$. Our description here simplifies the one given in~\cite{Fred82} on the said algorithm. Second, we exclude $t=0$ since the resulting $\cG_f$ does not satisfy the two conditions in Theorem \ref{thm:po}, although, with the initial state $\1^{n}$, the GPO Algorithm generates the same de Bruijn sequence as the one produced by {\tt Prefer-Zero}.
\end{remark}

\begin{example}
\begin{figure}[h!]
	\centering
	\begin{tikzpicture}
	[
	> = stealth,
	shorten > = 2pt,
	auto,
	node distance = 1.4cm and 1cm,
	semithick
	]
	
	\tikzstyle{every state}=
	\node[rectangle,fill=white,draw,rounded corners,minimum size = 4mm]
	
	\node[state] (1) {$0001$};
	\node[state] (2) [right of=1] {$0000$};
	\node[state] (3) [left of=1] {$1000$};
	\node[state] (4) [below of=1] {$0011$};
	\node[state] (5) [left of=4] {$1001$};
	
	\node[state] (6) [left of=5] {$0100$};
	\node[state] (7) [below of=5] {$1100$};
	
	\node[state,fill=gray] (8) [below of=4] {$0111$};
	\node[state,fill=gray] (9) [right of=8] {$1011$};
	\node[state] (10) [right of=9] {$0101$};
	\node[state] (11) [right of=10] {$1010$};
	\node[state] (12) [above of=10] {$0010$};
	
	\node[state,fill=gray] (13) [below of=8] {$1110$};
	\node[state] (14) [left of=13] {$1111$};
	\node[state,fill=gray] (15) [right of=13] {$1101$};
	\node[state] (16) [right of=15] {$0110$};
	
	\path[->] (2) edge (1);
	\path[->] (3) edge (1);
	
	\path[->] (1) edge (4);
	\path[->] (5) edge (4);
	\path[->] (6) edge (5);
	\path[->] (7) edge (5);
	
	\path[->] (4) edge (8);
	\path[->] (8) edge (13);
	
	\path[->] (11) edge (10);
	\path[->] (12) edge (10);
	
	\path[->] (10) edge (9);
	\path[->] (15) edge (9);
	
	\path[->] (9) edge (8);

	\path[->] (16) edge (15);
	\path[->] (13) edge (15);
	\path[->] (14) edge (13);
	\end{tikzpicture}
	\caption{The state graph $\cG_f$ for $f=x_1 \cdot x_2 \cdot x_3 + 1$
		(\ie, $t=1$).}
	\label{fig:SG1}
\end{figure}
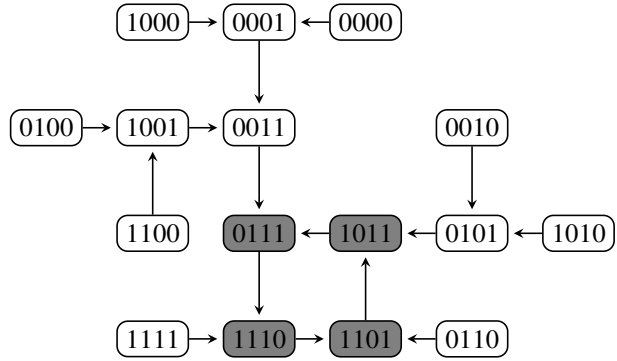
Figure~\ref{fig:SG1} presents $\cG_f$ for $f=x_1 \cdot x_2 \cdot x_3 + 1$. Running Algorithm~\ref{algo:po} on input $f$ and $\bb=1110$ yields the following sequence of states
\begin{align*}
\bb = & 1110 \to 1100 \to 1000 \to 0000 \to \colorbox[gray]{0.8}{0001} \to 0010 
\to ~ 0100 \to \colorbox[gray]{0.8}{1001} \to \colorbox[gray]{0.8}{0011} \\
\to &~0110 \to \colorbox[gray]{0.8}{1101} \to 1010 
\to \colorbox[gray]{0.8}{0101} \to \colorbox[gray]{0.8}{1011} \to
\colorbox[gray]{0.8}{0111} \to 1111 \to \bb.
\end{align*}
The states in gray follow the rule of assignment in Line~\ref{Algpo:line7} of Algorithm~\ref{algo:po} since the respective preferred states designated by Line~\ref{Algpo:line5} have appeared earlier. The resulting de Bruijn sequence is $(1110~0001~0011~0101)$. \QEDB
\end{example}

\subsection{Family Three}
For order $n \geq 4$ we have yet another family of de Bruijn sequences that follows from the GPO Algorithm.

\begin{theorem}\label{thm:sp}
Given $n\geq4$, let the feedback function be
\begin{equation}\label{equ-sp}
f(x_0,x_1,\ldots,x_{n-1})=1 + x_{n-3} + x_{n-1} + x_{n-3} \cdot x_{n-2} 
 + x_{n-2} \cdot x_{n-1} + x_{n-3} \cdot x_{n-2} \cdot x_{n-1}
\end{equation}
and the initial state $\bb$ be any $n$-stage state of the sequence $(1110)$. Then the GPO Algorithm generates the $\mathcal{F}_3(n)$ family of de Bruijn sequences.
\end{theorem}

\begin{proof}
The state graph $\cG_f$ of the FSR with feedback function $f$ in Equation~\ref{equ-sp} also satisfies the conditions in Theorem \ref{thm:po}.
The coefficient of $x_0$ in $f$ is $0$, implying that each non-leaf state has two children. The graph $\cG_f$ contains a cycle whose vertices are the $n$-stage states of the periodic sequence $(1,1,1,0)$. All other $n$-stage states are divided into disjoint trees whose respective roots are the vertices of the cycle.
\end{proof}

\begin{example}
The state graph $\cG_f$ for $n=4$ and $f$ taken from Equation (\ref{equ-sp}) is in Figure~\ref{fig:SG2}.

\begin{figure}[h]
	\centering 
	\begin{tikzpicture}
[
> = stealth,
shorten > = 2pt,
auto,
node distance = 1.4cm and 1cm,
semithick
]

\tikzstyle{every state}=
\node[rectangle,fill=white,draw,rounded corners,minimum size = 4mm]
	
	\node[state] (1) {$0000$};
	\node[state] (2) [right of=1] {$0001$};
	\node[state] (3) [right of=2] {$1000$};
	\node[state] (4) [right of=3] {$0100$};
	\node[state] (5) [below of=2] {$0010$};
	
	\node[state] (6) [left of=5] {$1001$};
	\node[state] (7) [below of=3] {$1100$};
	
	\node[state] (8) [below of=5] {$0101$};
	\node[state] (9) [left of=8] {$1010$};
	\node[state,fill=gray] (10) [right of=8] {$1011$};
	\node[state,fill=gray] (11) [right of=10] {$0111$};
	\node[state] (12) [right of=11] {$0011$};
	
	\node[state,fill=gray] (13) [below of=11] {$1110$};
	\node[state] (14) [right of=13] {$1111$};
	\node[state,fill=gray] (15) [below of=10] {$1101$};
	\node[state] (16) [left of=15] {$0110$};
	
	\path[->] (1) edge (2);
	\path[->] (3) edge (2);
	
	\path[->] (4) edge (3);
	\path[->] (7) edge (3);
	\path[->] (2) edge (5);
	\path[->] (6) edge (5);
	
	\path[->] (5) edge (8);
	\path[->] (9) edge (8);
	
	\path[->] (8) edge (10);
	\path[->] (10) edge (11);
	
	\path[->] (15) edge (10);
	\path[->] (12) edge (11);
	
	\path[->] (11) edge (13);
	
	\path[->] (14) edge (13);
	\path[->] (13) edge (15);
	\path[->] (16) edge (15);
	\end{tikzpicture}
	\caption{The graph $\cG_f$ for $f = 1 + x_1 + x_3+ x_1 \cdot x_2 + x_2 \cdot x_3 + x_1 \cdot x_2 \cdot x_3$.}
	\label{fig:SG2}
\end{figure}
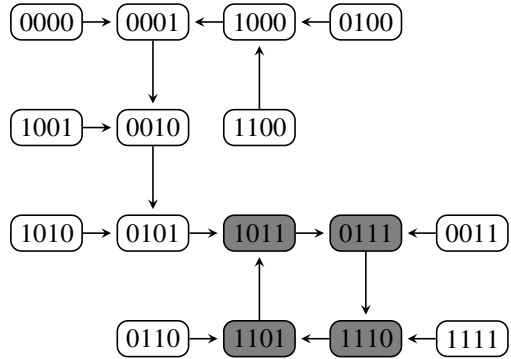

Starting from $\bb=1101$ yields the sequence of states
\begin{align*}
\bb = & 1101 \to 1010 \to 0100 \to 1001 \to 0011 \to 0110 
\to 1100 \to \colorbox[gray]{0.8}{1000} \to 0000 \\
\to &~\colorbox[gray]{0.8}{0001} \to \colorbox[gray]{0.8}{0010} \to \colorbox[gray]{0.8}{0101} 
\to \colorbox[gray]{0.8}{1011} \to \colorbox[gray]{0.8}{0111} \to 1111 \to \colorbox[gray]{0.8}{1110} \to \bb.
\end{align*}
As before, the states in gray are governed by Line~\ref{Algpo:line7} of Algorithm~\ref{algo:po} since the respective preferred states have already appeared. The output is $(1101~0011~0000~1011)$. \QEDB
\end{example}

We have discussed three families of de Bruijn sequences that the GPO Algorithm can construct by carefully selecting the feedback functions and the initial states that combine to satisfy the requirements in Theorem \ref{thm:po}. Table~\ref{table:more} lists four more verified families of feedback functions with $\bb=\0^n$ for $ n \geq 5$. The proof, for each family, follows the usual line of arguments to check that the two conditions in Theorem~\ref{thm:po} are met. The details are omitted here for brevity. Many more combinations of feedback functions and initial states $\bb$ are still there to be discovered. We leave this as an open direction to pursue. 

\begin{table}[h!]
\caption{More Families of Feedback Functions with $\bb=\0^n$ for GPO.}
\label{table:more}
\renewcommand{\arraystretch}{1.1}
\centering
\begin{tabular}{c l l }
		\hline
		No. & $f$ & Condition(s) \\
		\hline
		
		$1$ & $x_{n-1} + x_t \cdot x_{n-1}$ & $0 < t < n-1$ for $n \geq 3$ \\
		
		$2$ & $x_{n-1} + \displaystyle{\prod_{t=1}^{n-1} x_t }$ & $n \geq 3$\\
		
		$3$ & $x_{t} \cdot x_{t+1} + x_{t+1} \cdot x_{n-1}$ & $0 < t < n-2$ for $n \geq 4$ \\
		
		$4$ & $x_k \cdot x_{n-1} + x_{\ell} \cdot x_{n-1}$ & $0 < k < \ell < n-1$ for $n \geq 4$ \\
\hline
\end{tabular}
\end{table}

\section{Modifications of the GPO Algorithm}\label{sec:modify}

The GPO Algorithm is not guaranteed to generate de Bruijn sequences when the  pair $(f,\bb)$ fails to satisfy the conditions in Theorem \ref{thm:po}. This section shows that some modifications to the GPO Algorithm may turn such a pair into a viable choice. 

\subsection{Prefer No}

For a given $n$ and $t \in \bbra{n-1}$ let the feedback function be
\begin{equation}\label{eq:no}
f(x_0,x_1,\ldots,x_{n-1}):= \prod_{j=t}^{n-1} x_j = x_t \cdot x_{t+1} \cdot \ldots \cdot x_{n-1}.
\end{equation}
Note that $\cG_f$ with $f$ in Equation (\ref{eq:no}) does not satisfy the conditions in Theorem \ref{thm:po} since $\cG_f$ is not connected, being divided into two trees with respective roots $\0^n$ and $\1^n$. We modify the GPO Algorithm into the {\tt Prefer-No} Algorithm, presented here as Algorithm~\ref{algo:NO}, and prove that it generates a de Bruijn sequence for each $0 \leq t <n$. When $t=n-1$, the resulting sequence is the same with the one produced by {\tt Prefer-Opposite} in~\cite{Alh10}. Taking $t \in \{0,1\}$ produces the same sequence as the output of {\tt Prefer-One}. 

\begin{algorithm}[ht!]
\caption{{\tt Prefer-No}}
\label{algo:NO}
\begin{algorithmic}[1]
\renewcommand{\algorithmicrequire}{\textbf{Input:}}
\renewcommand{\algorithmicensure}{\textbf{Output:}}
\Require $n$ and $t$ with $1 \leq t < n$.
\Ensure A de Bruijn sequence of order $n$.
\State{$\bc=c_0, c_1, \ldots, c_{n-1} \gets \0^n$}	
\Do
\State{Print($c_0$)}
\If{$\bc = 0\1^{n-1}$}
\State{$\bc \gets \1^{n}$} \label{Alg2Line5}
\Else
\State{$y \gets c_t \cdot c_{t+1} \cdot \ldots \cdot c_{n-1}$}
\If{$c_1,c_2,\ldots,c_{n-1}, \overline{y}$ has not appeared before}
\State{$\bc \gets c_1,c_2, \ldots, c_{n-1} ,\overline{y}$} \label{Alg2Line8}
\Else
\State{$\bc \gets c_1, c_2, \ldots c_{n-1}, y$} \label{Alg2Line10}
\EndIf
\EndIf	
\doWhile {$\bc \neq \0^n$}
\end{algorithmic}
\end{algorithm}

\begin{theorem}
The {\tt Prefer-No} Algorithm generates a de Bruijn sequence for each tuple $(t,n)$ satisfying $ n \geq 2$ and $ 1 \leq t < n$. We call the resulting set of sequences the family $\mathcal{F}_4(n)$.
\end{theorem}

\begin{proof}
Using the same method as in the proof of Theorem~\ref{thm:po} we can confirm that the first state that Algorithm \ref{algo:NO} visits twice must be $\0^n$. 

Now, note that Algorithm~\ref{algo:NO} does not terminate at $1\0^{n-1}$ since it can still proceed to $\0^n$. For the same reason, it does not terminate at $0\1^{n-1}$ or $\1^n$. Any other state, say $\bu$, has two possible successors. At least one of them must have not been visited yet since their two possible predecessors include $\bu$. Thus, Algorithm~\ref{algo:NO} will not stop until it reaches $\0^n$ the second time around.
	
Next, we show that the algorithm visits every state. Before terminating, the algorithm must have visited $1\0^{n-1}$. Each non-leaf descendant of $1\0^{n-1}$ has two children. By a repeated application of Lemma~\ref{lemma1} we confirm that the algorithm must have visited $\0^n$ and all of its descendants. These include $\1^{n-1}0$, which is the successor of $\1^n$. Hence, the algorithm must have visited both $\1^n$ and $0\1^{n-1}$ by following its rule. Next, each non-leaf descendant of $0\1^{n-1}$ has two predecessors. Again, we use Lemma \ref{lemma1} to verify that the algorithm visits all of the descendants of $0\1^{n-1}$. Thus, all possible states are covered and the resulting sequence is indeed de Bruijn.
\end{proof}

\begin{example}
Figure~\ref{fig:SG3} gives the state graph $\cG_f$ with $f=x_2 \cdot x_3$, obtained by letting $n=4$ and $t=2$ in Equation (\ref{eq:no}).
\begin{figure}[h]
	\centering 
	\begin{tikzpicture}
[
> = stealth,
shorten > = 2pt,
auto,
node distance = 1.4cm and 1cm,
semithick
]

\tikzstyle{every state}=
\node[rectangle,fill=white,draw,rounded corners,minimum size = 4mm]
	
	\node[state] (1) {$0011$};
	\node[state] (3) [below of=1] {$0111$};
	\node[state] (2) [left of=3] {$1011$};
	\node[state] (4) [below of=3] {$1111$};
	
	\node[state] (5) [right of=1] {$0001$};
	\node[state] (9) [right of=5] {$0010$};
	\node[state] (6) [right of=9] {$1001$};
	\node[state] (7) [right of=6] {$0101$};
	\node[state] (10) [below of=6] {$1010$};
	\node[state] (8) [right of=10] {$1101$};

	\node[state] (11) [below of=9] {$0100$};
	\node[state] (15) [below of=11] {$1000$};
	\node[state] (12) [right of=15] {$1100$};
	
	\node[state] (13) [right of=12] {$0110$};
	\node[state] (14) [right of=13] {$1110$};
	
	\node[state] (16) [left of=15] {$0000$};
	
	\path[->] (1) edge (3);
	\path[->] (2) edge (3);
	
	\path[->] (3) edge (4);
	\path[->] (4) edge[loop left] (4);
	
	\path[->] (5) edge (9);
	\path[->] (6) edge (9);
	
	\path[->] (7) edge (10);
	\path[->] (8) edge (10);
	
	\path[->] (9) edge (11);
	\path[->] (10) edge (11);
	
	\path[->] (13) edge (12);
	\path[->] (14) edge[bend right] (12);
	
	\path[->] (11) edge (15);
	\path[->] (12) edge (15);
	
	\path[->] (15) edge (16);
	\path[->] (16) edge[loop above] (16);
	\end{tikzpicture}
	\caption{The graph $\cG_f$ of FSR with $f=x_2 \cdot x_3$ and $n=4$, \ie, $t=2$}
	\label{fig:SG3}
\end{figure}
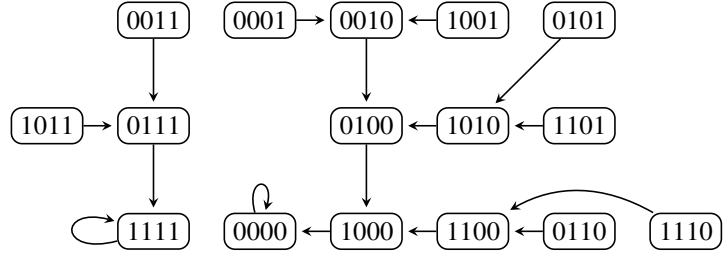
Running Algorithm~\ref{algo:NO} produces the sequence of states
\begin{align*}
\bb &= 0000 \to 0001 \to 0011 \to 0110 \to 1101 \to 1011 
\to \colorbox[gray]{0.8}{0111} \to \underline{1111} \to 1110 \\
\to &~\colorbox[gray]{0.8}{1100} \to 1001 \to \colorbox[gray]{0.8}{0010}
\to 0101 \to \colorbox[gray]{0.8}{1010} \to \colorbox[gray]{0.8}{0100} \to \colorbox[gray]{0.8}{1000} \to \bb.
\end{align*}
The underlined state $1111$ is based on the rule in Line~\ref{Alg2Line5}. The entries in gray follows the rule in Line~\ref{Alg2Line10}. Taking $n=6$, one obtains the $5$ distinct de Bruijn sequences in Table~\ref{table:PreferNo}. \QEDB

\begin{table}[h!]
\caption{The family $\mathcal{F}_4(6)$ obtained from {\tt Prefer-No}.}
\label{table:PreferNo}
\renewcommand{\arraystretch}{1.1}
\centering
\begin{tabular}{cc}
	\hline
	$t$ & The Resulting de Bruijn Sequence \\
	\hline
	$1$ & $(0000 0010 ~ 1010 0101 ~ 1010 1110 ~ 1000 1001 ~ 10110010 ~ 00011000 ~ 11100111 ~ 10111111)$ \\
	
	$2$ & $(00000011 ~ 01101011 ~ 00110001 ~ 11011100 ~ 10111101 ~ 00111111 ~ 00001010 ~ 10001001)$ \\
	
	$3$ & $(00000011 ~ 10111001 ~ 11101011 ~ 11110001 ~ 10110100 ~ 11001011 ~ 00001010 ~ 10001001)$ \\
	
	$4$ & $(00000011 ~ 11011111 ~ 10011101 ~ 01110001 ~ 10110100 ~ 11001011 ~ 00001010 ~ 10001001)$ \\
	
	$5$ & $(00000011 ~ 11110111 ~ 10011101 ~ 01110001 ~ 10110100 ~ 11001011 ~ 00001010 ~ 10001001)$ \\
	\hline
\end{tabular}
\end{table}
\end{example}

\subsection{From Primitive Polynomials}
We can use any primitive polynomial as the main ingredient. Given a primitive polynomial 
$g(x)=1 + a_1 x + \ldots +a_{m-1} x^{m-1} + x^m \in \F_2[x]$ of degree $m < n$, let 
\begin{equation}\label{equ-f}
f(x_0,x_1,\ldots,x_{n-1}):= x_{n-m} + a_1 x_{n-m+1} + \ldots + a_{m-1} x_{n-1}.
\end{equation}
The state graph $\cG_f$ with $f$   given in Equation (\ref{equ-f}) is divided into two disjoint components. For example, the state graph for $f(x_0,x_1,x_2,x_3)=x_1+x_2$ is in Figure~\ref{fig:SG4}.

\begin{figure}[h]
\centering 
\begin{tikzpicture}
[
> = stealth,
shorten > = 2pt,
auto,
node distance = 1.4cm and 1cm,
semithick
]
	
\tikzstyle{every state}=
\node[rectangle,fill=white,draw,rounded corners,minimum size = 4mm]

\node[state] (1) {$1000$};
\node[state] (2) [below of=1] {$0000$};
	
\node[state] (3) [right of=1] {$1111$};
\node[state,fill=gray] (4) [below of=3] {$1110$};
	
\node[state,fill=gray] (5) [below of=4] {$1100$};
\node[state] (6) [below of=5] {$0110$};
\node[state] (7) [right of=3] {$0011$};
\node[state,fill=gray] (8) [below of=7] {$0111$};
\node[state,fill=gray] (9) [below of=8] {$1001$};
\node[state] (10) [below of=9] {$0100$};
		
\node[state] (11) [right of=7] {$1101$};
\node[state,fill=gray] (12) [below of=11] {$1011$};
\node[state,fill=gray] (13) [below of=12] {$0010$};
\node[state] (14) [below of=13] {$0001$};
\node[state,fill=gray] (15) [right of=13] {$0101$};
\node[state] (16) [below of=15] {$1010$};
	
\path[->] (1) edge (2);
\path[->] (2) edge[loop below] (2);
	
\path[->] (3) edge (4);
\path[->] (4) edge (5);
\path[->] (6) edge (5);
\path[->] (5) edge (9);
\path[->] (10) edge (9);
\path[->] (9) edge (13);
\path[->] (14) edge (13);

\path[->] (13) edge (15);
	
\path[->] (16) edge (15);
\path[->] (15) edge (12);
	
\path[->] (11) edge (12);
\path[->] (12) edge (8);
	
\path[->] (7) edge (8);
\path[->] (8) edge (4);
\end{tikzpicture}
\caption{The graph $\cG_f$ for $n=4$ and $f=x_1+x_2$ for $n=4$}
\label{fig:SG4}
\end{figure}
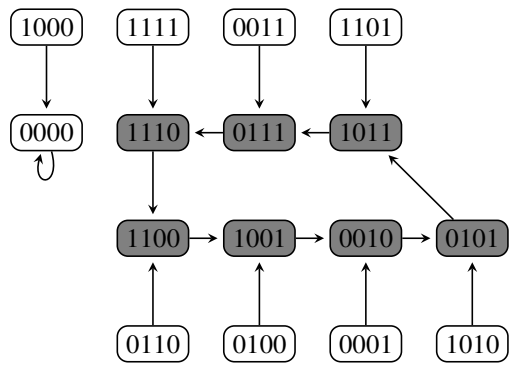

The {\tt Prim-Poly} Algorithm, presented here as Algorithm~\ref{algo:prim}, generalizes the example by modifying the GPO Algorithm. It produces another family, which we call $\mathcal{F}_5(n)$, of de Bruijn sequences by using the function in Equation~\ref{equ-f} and suitable initial states when $g(x)$ ranges over all primitive polynomials of degrees $1$ to $n-1$. The construction  in~\cite{WWZ18} is the special case when $g(x)=1+x+x^2$.

\begin{algorithm}[h]
\caption{{\tt Prim-Poly}}
\label{algo:prim}
\begin{algorithmic}[1]
\renewcommand{\algorithmicrequire}{\textbf{Input:}}
\renewcommand{\algorithmicensure}{\textbf{Output:}}
\Require $n$ and a degree $1 \leq m < n$ primitive polynomial
$g(x) = 1 + a_1 x + \ldots + a_{m-1} x^{m-1} + x^m$.
\Ensure de Bruijn sequences of order $n$.
\State{$\bm \gets $ the $m$-sequence generated by $g(x)$} \Comment{Repeat until the length reaches $n$}
\State{$\bc =c_0,c_1, \ldots, c_{n-1} \gets $ any arbitrary string of length $n$ in $\bm $} \Comment{Enumerate them}\label{pp:line2}
\State{$\bb=b_0,b_1,\ldots, b_{n-1} \gets \bc $}
\Do
\State{Print($c_0$)}
\If{$\bc =1\0^{n-1}$}\label{pp:line6}
	\State{$\bc \gets \0^{n}$}\label{pp:line7}
\Else
	\State{$y \gets c_{n-m} + a_1 c_{n-m+1}+ \ldots + a_{m-1} c_{n-1}$}
	\If{$c_1, c_2, \ldots, c_{n-1}, \overline{y}$ has not appeared before}
	\State{$\bc \gets c_1, c_2, \ldots, c_{n-1}, \overline{y}$}
	\Else
	\State{$\bc \gets c_1, c_2, \ldots, c_{n-1}, y$}
	\EndIf
\EndIf	
\doWhile{$\bc \neq \bb$}
\end{algorithmic}
\end{algorithm}

\begin{theorem}\label{thm:prim}
Let $g(x)$ be a given primitive polynomial of degree $m$. Algorithm \ref{algo:prim} generates de Bruijn sequences for $1 \leq m < n$ from the feedback function $f$ in Equation (\ref{equ-f}).
\end{theorem}

\begin{proof}
Suppose that Algorithm \ref{algo:prim} takes an initial state $\bb$, which is a nonzero $n$-stage state of the $m$-sequence $\bm$ with primitive characteristic polynomial $g(x)$ of degree $m < n$. We can use the similar argument to the proof of Theorem~\ref{thm:po} to confirm that $\bb$ must be the first state to be visited twice.

The algorithm's rule of assignment forces $\0^{n-1}1$ to be the successor of $\0^n$, which in turn is the successor of $1\0^{n-1}$. Hence, the algorithm does not terminate at $\0^n$. Any other state $\bv \neq \0^n$ has two possible successors. When the algorithm visits $\bv$, there remains at least one of $\bv$'s successors to visit, allowing the algorithm to proceed further. The algorithm does not stop until it reaches $\bb$ the second time around.
	
Once the algorithm visits $\bb$ the second time around, it must have visited both of $\bb$'s children in $\cG_f$ by Lemma \ref{lemma1}. Any non-leaf state in the tree $T_{f,\bb}$ has two children. Lemma~\ref{lemma1} ensures that the algorithm has visited all such states as well as their descendants, including $\0^{n-1}1$. To visit $\0^{n-1}1$ the algorithm must have visited $\0^n$ and, prior to that, $1\0^{n-1}$. Since any non-leaf descendant of $1\0^{n-1}$ has two children, again by Lemma~\ref{lemma1}, the algorithm must have covered both $\0^n$ and $1\0^{n-1}$ as well as all of the descendants of $1\0^{n-1}$. Before reaching $\bb$ for the second time, the algorithm must have visited all possible states. Thus, the resulting sequence is de Bruijn.
\end{proof}

\begin{example}\label{ex:primpoly}
Let $n=6$ and $1 \leq m \leq 5$. There are $12$ choices for the primitive polynomial $g(x)$. There are $228$ distinct de Bruijn sequences of order $6$ that Algorithm \ref{algo:prim} can construct. Table~\ref{table:SPFCN} presents only one example for each primitive polynomial $g(x)$ of degree $m$. To save space, $g(x)$ is represented by the coefficients of its monomials in \underline{descending} degree, \eg, $1101$ represents $g(x)=x^3+x^2+1$. \QEDB

\begin{table}[h!]
\caption{Some de Bruijn Sequences in $\mathcal{F}_5(6)$: one from each primitive polynomial $g(x)$ of degree $1 \leq m \leq 5$.}
\label{table:SPFCN}
\renewcommand{\arraystretch}{1.1}
\centering
\small
\begin{tabular}{c r c}
\hline
$m$ & $g(x)$ & The Resulting de Bruijn Sequence \\
\hline
$1$ & $ 11 $ & $(0000 0011 ~ 1111 0101 ~ 0110 1001 ~ 0100 0101 ~ 1101 1001 ~ 0011 0111 ~ 1001 1100 ~ 0110 0001)$\\

$2$ & $ 111 $ & $(0000 0011 ~ 0101 0110 ~ 1111 1100 ~ 1001 0100 ~ 1111 0100 ~ 0100 0010 ~ 1110 0011 ~ 1011 0011)$\\

$3$ & $ 1011 $ & $(0000 0010 ~ 10111001 ~ 00101101 ~ 00011011 ~ 00010000 ~ 11111101 ~ 01001100 ~ 11101111)$\\

& $ 1101 $ & $(00000011 ~ 11011101 ~ 00100111 ~ 00010110 ~ 00010001 ~ 10010101 ~ 11111001 ~ 10110101)$\\

$4$ & $ 10011$ & $(00000010 ~ 11011100 ~ 11111101 ~ 00100110 ~ 10101111 ~ 00010000 ~ 11101100 ~ 10100011)$\\
 & $ 11001$ & $(00000011 ~ 00010010 ~ 11111100 ~ 11110101 ~ 01100100 ~ 01110000 ~ 10100110 ~ 11101101)$\\

$5$ & $100101$ & $(00000011 ~ 10010001 ~ 00101100 ~ 11111100 ~ 01101110 ~ 10100001 ~ 01011110 ~ 11010011)$\\

& $101001$ & $(00000011 ~ 00101101 ~ 11101010 ~ 00101011 ~ 10110001 ~ 11111001 ~ 10100100 ~ 00100111)$\\
& $101111$ & $(00000011 ~ 01100111 ~ 10100101 ~ 01110001 ~ 01101010 ~ 00111011 ~ 11110010 ~ 01100001)$\\
& $110111$ & $(00000010 ~ 11101101 ~ 01001111 ~ 00011100 ~ 11011111 ~ 10100010 ~ 01010110 ~ 00011001)$\\
& $111011$ & $(00000010 ~ 01100011 ~ 01010010 ~ 00101111 ~ 11011001 ~ 11000011 ~ 11001010 ~ 11011101)$\\
& $111101$ & $(00000010 ~ 00110010 ~ 01111110 ~ 11100010 ~ 10110100 ~ 00111010 ~ 10010111 ~ 10011011)$\\
\hline
\end{tabular}
\end{table}
\end{example}

Based on distinct initial states, almost all of the de Bruijn sequences generated by Algorithm \ref{algo:prim} over all degree $m$ primitive polynomials for $1 \leq m < n$ are distinct. Computational evidences suggest that the only exemption occurs for each primitive polynomial $g(x)$ of degree $m=n-1$ where there is exactly one pair of distinct states, namely $0\1^{n-1}$ and $\1^{n-1}0$, that yields two cyclically equivalent de Bruijn sequences.

Let $\varphi$ be Euler's totient function. There are $\frac{\varphi \left(2^m-1 \right)}{m}$ primitive binary polynomials of degree $m \geq 1$. For $1 \leq m \leq n-2$ the number of distinct initial states from the $m$-sequence $\bm$ is its period $2^m-1$. Distinct initial sequences yield distinct de Bruijn sequences. When $2 < n = m+1$, there are only $2^{n-1}-2$ distinct de Bruijn sequences produced from the $2^{n-1}-1$ distinct initial states since $\1^{m}0$ and $0\1^{m}$ yield the same de Bruijn sequence. We formalize this observation and note that the proof follows the same reasoning as that of Theorem~\ref{thm:distinct}. The crux move is to show that there are no other collisions than the one already mentioned above. The details are omitted here.

\begin{theorem}\label{thm:distinctprim}
Let $n > 2$. Then
\[
\abs{\mathcal{F}_5(n)} = \frac{\varphi \left(2^{n-1}-1 \right)}{n-1} (2^{n-1}-2) + \sum_{m=1}^{n-2} \frac{\varphi \left(2^m-1 \right)}{m} (2^m-1).
\] 	
\end{theorem}

Next, we describe the family $\mathcal{F}_6$, which is very closely related to $\mathcal{F}_5$. Again, let $g(x)$ ranges over all suitable primitive polynomials $ 1 + a_1 x + \ldots + a_{m-1} x^{m-1} + x^m \in \F_2[x]$ of degree $1 \leq m < n$. Let $f(x_0,x_1,\ldots,x_{n-1})$ be the feedback function in Equation (\ref{equ-f}), defined by $g(x)$. We use the feedback function
\begin{equation}\label{eq:yes}
\widetilde{f}(x_0,x_1,\ldots,x_{n-1}) := f(x_0,x_1,\ldots,x_{n-1})+1 
= 1 + x_{n-m} + a_1 x_{n-m+1} + \ldots + a_{m-1} x_{n-1}
\end{equation}
as an input in a slightly modified Algorithm~\ref{algo:prim}. The initial state is now the \emph{complement} of any string of length $n$ in the $m$-sequence $\bm$. 
Consequently, in Line~\ref{pp:line6}, we use $\bc=0\1^{n-1}$ and, in Line~\ref{pp:line7} the assignment rule becomes $\bc \gets \1^n$. The respective analogues of Example~\ref{ex:primpoly} and Theorem~\ref{thm:distinctprim} also hold for $\mathcal{F}_6(n)$.

\subsection{Other Special Functions}\label{subsec:spfcn}

A strategy that works well is to choose special feedback functions whose respective corresponding state graphs are similar to the successful ones above. Here we provide another example. Such functions are far too numerous to list. We invite the readers to come up with their own favourites.

\begin{algorithm}[h]
\caption{A Special Function}
\label{algo:specialfcn}
\begin{algorithmic}[1]
\renewcommand{\algorithmicrequire}{\textbf{Input:}}
\renewcommand{\algorithmicensure}{\textbf{Output:}}
\Require the feedback function $f$ in Equation (\ref{equ-sff}).
\Ensure A de Bruijn sequence of order $n$.
\State{$\bc =c_0,c_1, \ldots, c_{n-1} \gets $ any $n$-stage state of the periodic sequence $(0111)$} \Comment{Enumerate them}
\State{$\bb=b_0,b_1,\ldots, b_{n-1} \gets \bc $}
\Do
	\State{Print($c_0$)}
	\If{$\bc =1,\0^{n-1}$}
		\State{$\bc \gets \0^{n}$}
	\Else
		\State{$y \gets c_{n-3} + c_{n-2}+ c_{n-3} \cdot c_{n-2} + c_{n-3} \cdot c_{n-2} \cdot c_{n-1}$}
		\If{$c_1, c_2, \ldots, c_{n-1}, \overline{y}$ has not appeared before}
			\State{$\bc \gets c_1, c_2, \ldots, c_{n-1}, \overline{y}$}
		\Else
			\State{$\bc \gets c_1, c_2, \ldots, c_{n-1}, y$}
		\EndIf
	\EndIf	
\doWhile{$\bc \neq \bb$}
\end{algorithmic}
\end{algorithm}

\begin{proposition}\label{prop:gp}
Let $n \geq 4$. We run Algorithm \ref{algo:specialfcn} on
\begin{equation}\label{equ-sff}
f(x_0,x_1,\ldots,x_{n-1}):= 
x_{n-3} + x_{n-2} + x_{n-3} \cdot x_{n-2} + x_{n-3} \cdot x_{n-2} \cdot x_{n-1}
\end{equation}
The initial state can be any of the $n$-stage state of the periodic sequence $(0111)$. Then the resulting sequence is de Bruijn of order $n$.
\end{proposition}

The proof of Proposition~\ref{prop:gp} follows the argument in the proof of Theorem \ref{thm:prim} since the state graph for $f$ in Equation (\ref{equ-sff}) is similar with the one in Theorem~\ref{thm:prim}.

\begin{example}
Figure~\ref{fig:SG5} presents $\cG_f$ with $n=4$ and $f$ in (\ref{equ-sff}).
\begin{figure}[ht]
\centering 
\begin{tikzpicture}
[
> = stealth,
shorten > = 2pt,
auto,
node distance = 1.4cm and 1cm,
semithick
]

\tikzstyle{every state}=
\node[rectangle,fill=white,draw,rounded corners,minimum size = 4mm]
	
	\node[state] (1) {$1000$};
	\node[state] (2) [below of=1] {$0000$};
	
	\node[state] (3) [right of=2] {$1010$};
	\node[state] (4) [right of=3] {$0101$};
	
	\node[state] (5) [above of=4] {$0100$};
	\node[state] (6) [right of=5] {$1001$};
	\node[state] (7) [right of=6] {$1100$};
	\node[state] (8) [below of=7] {$0001$};
	\node[state] (9) [below of=6] {$0010$};
	\node[state,fill=gray] (10) [below of=4] {$1011$};
	
	\node[state,fill=gray] (11) [right of=10] {$1101$};
	\node[state] (12) [right of=11] {$0110$};
	\node[state,fill=gray] (13) [below of=10] {$0111$};
	\node[state] (14) [left of=13] {$0011$};
	\node[state,fill=gray] (15) [right of=13] {$1110$};
	\node[state] (16) [right of=15] {$1111$};
	
	\path[->] (1) edge (2);
	\path[->] (2) edge[loop below] (2);
	
	\path[->] (3) edge (4);
	\path[->] (4) edge (10);
	\path[->] (5) edge (6);
	\path[->] (7) edge (6);
	\path[->] (6) edge (9);
	\path[->] (8) edge (9);
	\path[->] (9) edge (4);
	
	\path[->] (12) edge (11);
	
	\path[->] (11) edge (10);
	\path[->] (10) edge (13);
	
	\path[->] (14) edge (13);
	\path[->] (13) edge (15);
	
	\path[->] (16) edge (15);
	\path[->] (15) edge (11);
	\end{tikzpicture}
	\caption{The graph $\cG_f$ for $f=x_1 + x_2 + x_1 \cdot x_2 + x_1 \cdot x_2 \cdot x_3$ from Equation (\ref{equ-sff})}
	\label{fig:SG5}
\end{figure}
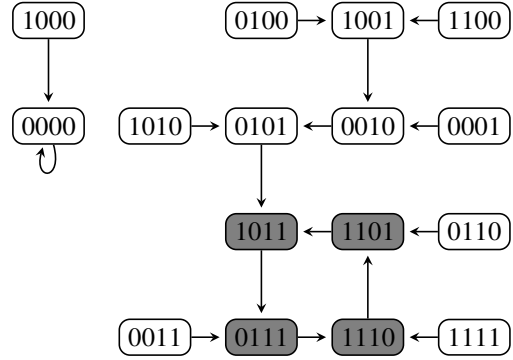
Starting from $\bb=1011$ we obtain the sequence of states
\begin{align*}
\bb = & 1011 \to 0110 \to 1100 \to 1000 \to \underline{0000} \to 0001 
\to 0011 \to \colorbox[gray]{0.8}{0111} \to  1111 \\
\to &~\colorbox[gray]{0.8}{1110} \to \colorbox[gray]{0.8}{1101} \to 1010 
\to 0100 \to \colorbox[gray]{0.8}{1001} \to \colorbox[gray]{0.8}{0010} \to \colorbox[gray]{0.8}{0101} \to \bb.
\end{align*}
\QEDB
\end{example}

\section{Conclusion}\label{sec:conclu}
We put many known greedy algorithms for binary de Bruijn sequences into the framework of GPO Algorithm. We establish easy-to-check conditions that guarantee the efficacy of this general algorithm and its variants. The Boolean maps, \ie, the truth table in its algebraic normal form (ANF) is used to output the resulting de Bruijn sequence. Efficiency in either time or memory complexity as $n$ grows is not a strong aspect of most greedy algorithms. Here, the bottleneck clearly lies in checking if a particular $n$-string has appeared. In our basic {\tt python} implementation, the search is iterated through the state graph. A typical modern laptop with $16$~GB of memory has the required resources to comfortably handle $n \leq 15$. Devising a better general search approach is interesting but lies outside the scope of this paper. 

A complete list of suitable feedback functions and initial states are easy to generate for small $n$. Table~\ref{table:GPO3} is for $n=3$. The combined Tables~\ref{table:GPO4} and~\ref{table:GPO4part2} is for $n=4$. They lead to the following assertion. For any $n >2$, let $\bS$ be any de Bruijn sequence of order $n$ and let $\bb$ be any $n$-string. Then there exists $f(x_0,x_1,\ldots,x_{n-1})$ such that $\bS$ is the output of the GPO Algorithm on input $(f,\bb)$. 

One defines a function $g(x_0,x_1,\ldots,x_{n-2})$ while traversing $\bS$, starting from the initial state $\bb=b_0,b_1,\ldots,b_{n-1}$. Since the root is special, $g(b_0,b_1,\ldots,b_{n-2}):=b_{n-1}$. If $y$ is the next bit after the \emph{first} occurance of any other $n-1$ string $(x_0,x_1,\ldots,x_{n-2})$, then $g(x_0,x_1,\ldots,x_{n-2}):=\overline{y}$. This function $g$ defines a suitable feedback function $f(x_0,x_1,\ldots,x_{n-1}) := g(x_1,\ldots,x_{n-1})$ to use in the GPO Algorithm. 

\begin{example}
Let $n=3$, $\bS=(0001~0111)$, and $\bb=101$. We have $g(1,0)=1$, $g(0,1)=0$, $g(1,1)=0$, and $g(0,0)=1$. Hence, $g(x_0,x_1) = 1 + x_1$ and $f(x_0,x_1,x_2) := g(x_1,x_2) = 1 + x_2$. On input $\bb=101$ and $f(x_0,x_1,x_2) =1+x_2$, the algorithm generates $\bS$.

Let $n=4$, $\bS=(1111~0110~0101~0000)$, and $\bb=1111$. We have 
\[
g(\bv)=1 \mbox{ for } \bv \in \{111,011,001,000\} \mbox{ and }
g(\bv)=0 \mbox{ for } \bv\in \{110,101,100,010\}.
\]
Hence, $g(x_0,x_1,x_2)=x_0 \cdot x_1 + x_0 + x_1 \cdot x_2 + x_1 + 1$ and, therefore, $f(x_0,x_1,x_2,x_3)=g(x_1,x_2,x_3) = x_1 \cdot x_2 + x_1 + x_2 \cdot x_3 + x_2 + 1$. \QEDB
\end{example}

\begin{table}[h]
	\caption{GPO List for $n = 3$}
	\label{table:GPO3}
	\renewcommand{\arraystretch}{1.1}
	\centering
	\begin{tabular}{c  c c | c c}
		\hline
		No. & $f$ for $\bS=(0001~0111)$  & $\bb \in $ & $f$ for $\bS=(0001~1101)$ & $\bb \in $\\
		\hline
		$1$ & $x_1 \cdot x_2 + x_1 + x_2 + 1$ & $\{010\}$ & $x_1 \cdot x_2 + x_1 + x_2 + 1$ & $\{100, 001\}$ \\
		
		$2$ &  $x_1 + 1$ & $\{100,001\}$ & $x_2 \cdot x_1 + x_1 + 1$ & $\{111\}$\\
		$3$ & $x_2 + 1$ & $\{101\}$ & $x_1 + 1$ & $\{011,110\}$ \\
		$4$ & $x_2 \cdot x_1 + 1$ & $\{011,110\}$  & $x_2 + 1$ & $\{010\}$\\
		$5$ & $x_2 \cdot x_1 + x_2$ & $\{000\}$  & $x_2 \cdot x_1 + 1$ & $\{101\}$ \\
		$6$ & $1$ & $\{111\}$  & $0$ & $\{000\}$ \\
		\hline
	\end{tabular}
\end{table}

\section*{Acknowledgements}
The work of Z.~Chang is supported by the National Natural Science Foundation of China under Grant 61772476 and the Key Scientific Research Projects of Colleges and Universities in Henan Province under Grant 18A110029. Nanyang Technological University Grant Number M4080456 supports the research carried out by M.~F.~Ezerman and A.~A.~Fahreza.

\section*{References}

\appendix

\section{Proof of Theorem~\ref{thm:distinct}}

To make the proof easier to follow we start with an illustration.

Let $m=3$ and $h(x_0,x_1,x_2)=1 + x_0 + x_1 \cdot x_2 + x_2$. Figure~\ref{fig:M3} is the state graph $\cG_h$ corresponding to the FSR that produces the de Bruijn sequence $\bS_3:=(1100~0101)$. Figures~\ref{fig:M4} and~\ref{fig:M5} give, respectively, the state graphs corresponding to the functions $f_4=1+ x_1 + x_2 \cdot x_3 + x_3$ and $f_5=1 + x_2 + x_3 \cdot x_4 + x_4$.

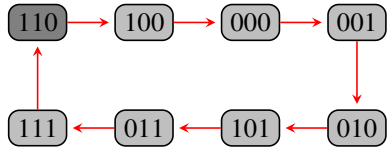
\begin{figure}[h]
	\centering
	\begin{tikzpicture}
	[
	> = stealth,
	shorten > = 2pt,
	auto,
	node distance = 1.4cm and 1cm,
	semithick
	]
	
	\tikzstyle{every state}=
	\node[rectangle,fill=white,draw,rounded corners,minimum size = 4mm]
	
	\node[state,fill=gray] (1) {$110$};
	\node[state,fill=lightgray] (2) [right of=1] {$100$};
	\node[state,fill=lightgray] (3) [right of=2] {$000$};
	\node[state,fill=lightgray] (4) [right of=3] {$001$};
	\node[state,fill=lightgray] (5) [below of=4] {$010$};
	\node[state,fill=lightgray] (6) [left of=5] {$101$};
	\node[state,fill=lightgray] (7) [left of=6] {$011$};
	\node[state,fill=lightgray] (8) [below of=1] {$111$};
	
	\path[red,->] (1) edge (2);
	\path[red,->] (2) edge (3);
	\path[red,->] (3) edge (4);
	\path[red,->] (4) edge (5);
	\path[red,->] (5) edge (6);
	\path[red,->] (6) edge (7);
	\path[red,->] (7) edge (8);
	\path[red,->] (8) edge (1);
	\end{tikzpicture}
	\caption{The state graph $\cG_{h_3}$ for $h_3=1+ x_0 + x_1 \cdot x_2 + x_2$.}
	\label{fig:M3}
\end{figure}

\begin{figure}[h]
	\centering
	\begin{tikzpicture}
	[
	> = stealth,
	shorten > = 2pt,
	auto,
	node distance = 1.5cm,
	semithick
	]
	
	\tikzstyle{every state}=
	\node[rectangle,fill=white,draw,rounded corners,minimum size = 4mm]
	
	\node[state,fill=gray] (1) {$\underline{1\overline{10}}\overline{0}$};
	\node[state,fill=lightgray] (2) [right of=1] {$\underline{1\overline{00}}\overline{0}$};
	\node[state,fill=lightgray] (3) [right of=2] {$\underline{0\overline{00}}\overline{1}$};
	\node[state,fill=lightgray] (4) [right of=3] {$\underline{0\overline{01}}\overline{0}$};
	\node[state,fill=lightgray] (5) [below of=4] {$\underline{0\overline{10}}\overline{1}$};
	\node[state,fill=lightgray] (6) [left of=5] {$\underline{1\overline{01}}\overline{1}$};
	\node[state,fill=lightgray] (7) [left of=6] {$\underline{0\overline{11}}\overline{1}$};
	\node[state,fill=lightgray] (8) [below of=1] {$\underline{1\overline{11}}\overline{0}$};
	
	\node[state] (9) [left of=1] {$0110$};
	\node[state] (10) [above of=2] {$0100$};
	\node[state] (11) [above of=3] {$0000$};
	\node[state] (12) [right of=4] {$1001$};
	\node[state] (13) [right of=5] {$1010$};
	\node[state] (14) [below of=6] {$1101$};
	\node[state] (15) [below of=7] {$0011$};
	\node[state] (16) [left of=8] {$1111$};
	
	\path[red,->] (1) edge (2);
	\path[red,->] (2) edge (3);
	\path[red,->] (3) edge (4);
	\path[red,->] (4) edge (5);
	\path[red,->] (5) edge (6);
	\path[red,->] (6) edge (7);
	\path[red,->] (7) edge (8);
	\path[red,->] (8) edge (1);
	
	\path[->] (9) edge (1);
	\path[->] (10) edge (2);
	\path[->] (11) edge (3);
	\path[->] (12) edge (4);
	\path[->] (13) edge (5);
	\path[->] (14) edge (6);
	\path[->] (15) edge (7);
	\path[->] (16) edge (8);
	\end{tikzpicture}
	\caption{The state graph $\cG_{f_4}$ for $f_4 = 1+ x_1 + x_2 \cdot x_3 + x_3$.}
	\label{fig:M4}
\end{figure}
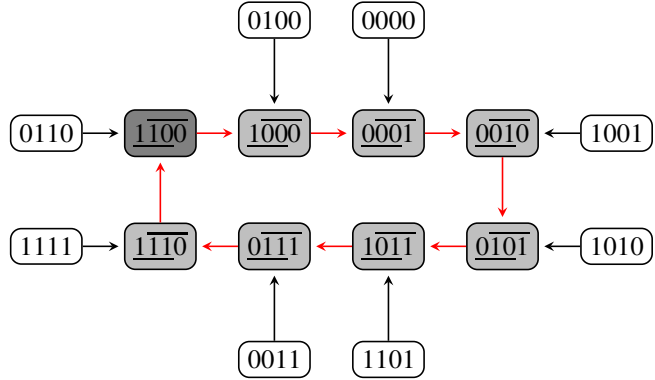

\begin{figure*}[ht!]
	\centering
	\begin{tikzpicture}
	[
	> = stealth,
	shorten > = 2pt,
	auto,
	node distance = 1.5cm,
	semithick
	]
	
	\tikzstyle{every state}=
	\node[rectangle,fill=white,draw,rounded corners,minimum size = 4mm]
	
	\node[state,fill=gray] (1) {$\underline{11\overline{0}}\overline{00}$};
	\node[state,fill=lightgray] (2) [right of=1] {$\underline{10\overline{0}}\overline{01}$};
	\node[state,fill=lightgray] (3) [right of=2] {$\underline{00\overline{0}}\overline{10}$};
	\node[state,fill=lightgray] (4) [right of=3] {$\underline{00\overline{1}}\overline{01}$};
	\node[state,fill=lightgray] (5) [below of=4] {$\underline{01\overline{0}}\overline{11}$};
	\node[state,fill=lightgray] (6) [left of=5] {$\underline{10\overline{1}}\overline{11}$};
	\node[state,fill=lightgray] (7) [left of=6] {$\underline{01\overline{1}}\overline{10}$};
	\node[state,fill=lightgray] (8) [below of=1] {$\underline{11\overline{1}}\overline{00}$};
	
	\node[state] (9) [left of=1] {$01100$};
	\node[state] (10) [above of=2] {$01000$};
	\node[state] (11) [above of=3] {$00001$};
	\node[state] (12) [right of=4] {$10010$};
	\node[state] (13) [right of=5] {$10101$};
	\node[state] (14) [below of=6] {$11011$};
	\node[state] (15) [below of=7] {$00111$};
	\node[state] (16) [left of=8] {$11110$};
	
	\node[state] (17) [left of=9] {$00110$};
	\node[state] (18) [above of=9] {$10110$};
	
	\node[state] (19) [left of=10] {$10100$};
	\node[state] (20) [above of=10] {$00100$};
	
	\node[state] (21) [above of=11] {$10000$};
	\node[state] (22) [right of=11] {$00000$};
	
	\node[state] (23) [above of=12] {$01001$};
	\node[state] (24) [right of=12] {$11001$};
	
	\node[state] (25) [right of=13] {$11010$};
	\node[state] (26) [below of=13] {$01010$};
	
	\node[state] (27) [right of=14] {$01101$};
	\node[state] (28) [below of=14] {$11101$};
	
	\node[state] (29) [below of=15] {$00011$};
	\node[state] (30) [left of=15] {$10011$};
	
	\node[state] (31) [below of=16] {$11111$};
	\node[state] (32) [left of=16] {$01111$};
	
	\path[red,->] (1) edge (2);
	\path[red,->] (2) edge (3);
	\path[red,->] (3) edge (4);
	\path[red,->] (4) edge (5);
	\path[red,->] (5) edge (6);
	\path[red,->] (6) edge (7);
	\path[red,->] (7) edge (8);
	\path[red,->] (8) edge (1);
	
	\path[->] (9) edge (1);
	\path[->] (10) edge (2);
	\path[->] (11) edge (3);
	\path[->] (12) edge (4);
	\path[->] (13) edge (5);
	\path[->] (14) edge (6);
	\path[->] (15) edge (7);
	\path[->] (16) edge (8);
	
	\path[->] (17) edge (9);
	\path[->] (18) edge (9);
	
	\path[->] (19) edge (10);
	\path[->] (20) edge (10);
	
	\path[->] (21) edge (11);
	\path[->] (22) edge (11);
	
	\path[->] (23) edge (12);
	\path[->] (24) edge (12);
	
	\path[->] (25) edge (13);
	\path[->] (26) edge (13);
	
	\path[->] (27) edge (14);
	\path[->] (28) edge (14);
	
	\path[->] (29) edge (15);
	\path[->] (30) edge (15);
	
	\path[->] (31) edge (16);
	\path[->] (32) edge (16);
	\end{tikzpicture}
	\caption{The state graph $\cG_{f_5}$ for $f_5 = 1+ x_2 + x_3 \cdot x_4 + x_4$.}
	\label{fig:M5}
\end{figure*}
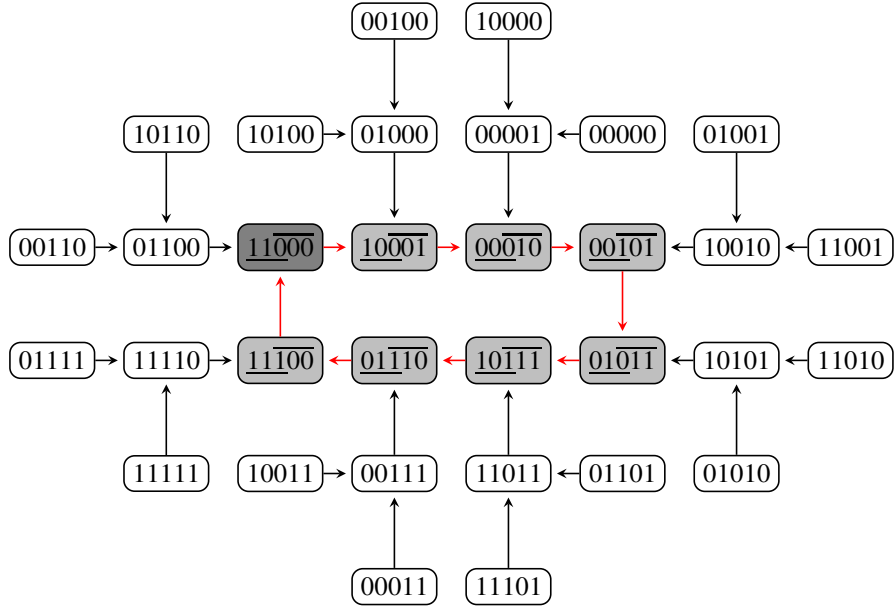

\begin{table}
	\caption{Order of appearance in a run of the GPO algorithm.}
	\label{table:order}
	\renewcommand{\arraystretch}{1.1}
	\centering
	\small
	\begin{tabular}{c c | c c | c c | c c}
		\hline
		\multicolumn{8}{c}{$\bS_5$ with $\bb=\bu_0=11000$} \\
		\hline
		No. & State & No. & State & No. & State & No. & State \\
		\hline
		$1$ & $\bu_0$ & $9$ & $10100 \in T_1$ & $17$ & $11011 \in T_5$ & $25$ & $\bu_2$\\
		
		$2$ & $\widehat{\bu_1} \in T_2$ & $10$ & $01001 \in T_3$ & $18$ & $\widehat{\bu_5} \in T_0$ & $26$ & $\bu_3$\\
		
		$3$ & $\0^5 \in T_2$ & $11$ & $10011 \in T_6$ & $19$ & $01100 \in T_0$ & $27$ & $\widehat{\bu_4} \in T_4$\\
		
		$4$ & $\0^41 \in T_2$ & $12$ & $00111 \in T_6$ & $20$ & $\widehat{\bu_0} \in T_3$ & $28$ & $10101 \in T_4$\\
		
		$5$ & $\widehat{\bu_2} \in T_6$ & $13$ & $\widehat{\bu_6} \in T_7$ & $21$ & $10010 \in T_3$ & $29$ & $\bu_4$\\
		
		$6$ & $00110 \in T_0$ & $14$ & $\1^5 \in T_7$ & $22$ & $\widehat{\bu_3} \in T_1$ & $30$ & $\bu_5$\\
		
		$7$ & $01101 \in T_5$ & $15$ & $\1^40 \in T_7$ & $23$ & $01000 \in T_1$ & $31$ & $\bu_6$\\
		
		$8$ & $11010 \in T_4$ & $16$ & $\widehat{\bu_7} \in T_5$ & $24$ & $\bu_1 $ & $32$ & $\bu_7$\\
		%
		%
		%
		%
		%
		%
		%
		\hline
	\end{tabular}
\end{table}

For fixed $n$ and $m$, with $n \geq m$ and $j \in \bbra{2^m-1}$, let $\bu_j$ denote the $2^m$ consecutive states that form a cycle $\cC_f$ in $\cG_f$. These states are typeset in gray in Figures~\ref{fig:M3} to~\ref{fig:M5}. As $j$ progresses from $0$ to $2^m-1$, we follow the directed edges to traverse all of the states clockwise from top left. For example, the state $\bu_2$ as shown in Figure~\ref{fig:M5} is $00010$. We underline the first $m=3$ bits of the states in $\cC_{f_4}$ and $\cC_{f_5}$ and put a line over their last $m=3$ bits.

For a fixed $n > m$, let $\cR_f$ be the graph obtained by removing the edges in $\cC_f$ from the state graph $\cG_f$. Note that $\cR_f$ has $2^m$ disjoint trees as components. Let $T_j$ denote the tree in $\cR_f$ that contains $\bu_{j}$ as its root, for $j \in \bbra{2^m-1}$. In Figure~\ref{fig:M5}, for instance, the vertex set of $T_0$ is $\{10110,00110,01100,11000\}$.

The algorithm, on input $f_4$ and initial state $\bu_0=1100$, yields $\bS_4:=(1100 ~1101~ 0000~ 1011)$. Note that the states in $\cC_{f_4} \subset \cG_{f_4}$ occur in the exact same order as they do in $\bS_4$. In one period the states appear in the order
\[
\bu_0 ,\, \widehat{\bu_1} ,\, \widehat{\bu_3} ,\, \widehat{\bu_6} ,\, \widehat{\bu_0} ,\, \widehat{\bu_5} ,\, \widehat{\bu_4} ,\, \bu_1 ,\, \widehat{\bu_2} ,\, \bu_2 ,\, \bu_3 ,\, \bu_4 ,\, \bu_5 ,\, \bu_6 ,\, \widehat{\bu_7} ,\, \bu_7.
\]

The same holds for $n=5$ with initial state $\bu_0=11000$. The states in $\cC_{f_5} \subset \cG_{f_5}$ follow the same order of appearance as they do in the resulting de Bruijn sequence 
\[
\bS_5:=(11000001~10100111~11011001~00010101).
\] 
The states appear in the order listed in Table~\ref{table:order}. 

Let $\bm_0,\bm_1,\ldots,\bm_{2^m-1}$ be the consecutive states of $\bS_m$. There is a natural bijection between $\bu_k$ and $\bm_k$ for $k \in \bbra{2^m-1}$ given by $\bu_k:=a,b,c,d,e \longleftrightarrow \bm_k:=a,b,c$. Hence, $\bu_0 = 11000 \longleftrightarrow \bm_0=110$ and so on. Slightly abusing the notation, we use $T_{\widehat{\bm_k}}$ to denote the tree in $\cR_f$ whose root is $T_{\widetilde{\bu_k}}$ with $\widehat{\bm_k}$ being the first $m=3$ bits of $\widetilde{\bu_k}$.

Let $V_{\ell}:=\{\bu_{\ell-1}\} \cup \{\bv \in T_{\ell} : \bv \neq \bu_{\ell}\}$, for $\ell \in \bbra{i+1,j}$, be the set that contains $\bu_{\ell-1}$ and the vertices of $T_{\ell}$ except for its root $\bu_{\ell}$. Notice that elements in $V_{\ell}$ always come in pairs as conjugate states. The first two states are conjugate as are the last two states. Their last $m=3$ bits are $\bm_{\ell}$ and $\bm_{\ell+1}$. This fact follows from how $f$ is defined in equation (\ref{equ-db}) where $h$ is modified by shifting the focus to the last $m$ entries, instead of the first $m$ entries. In Figure~\ref{fig:M5} we have
\begin{align*}
V_1 &=\{10100,00100,01000,11000\} \rightarrow \bm_1=100,\bm_2=000, \\
V_2 &=\{10001,00001,10000,00000\} \rightarrow \bm_2=000,\bm_3=001, \\
V_3 &=\{00010,10010,01001,11001\} \rightarrow \bm_3=001,\bm_4=010.
\end{align*}
One of the two states in each conjugate pair belonging to $V_{\ell}$ has a successor which is a leaf in either $T_{\widehat{\bm_{\ell}}}:=T_{\widetilde{\bu_{\ell}}}$ or $T_{\widehat{\bm_{\ell+1}}}:=T_{\widetilde{\bu_{\ell+1}}}$.
The successor of $10100$ is $01001$, which is a leaf in $T_{\widehat{000}}=T_{\widetilde{\bu_2}}=T_3$. The successor of $11000$ is $10000$, which is a leaf in $T_{\widehat{001}}=T_{\widetilde{\bu_3}}=T_2$.

We will use the enumeration of the trees that contribute some leaves to a well-chosen subsequence of a de Bruijn sequence as a tool in the proof of the next theorem. Figure~\ref{fig:M4} will be useful to confirm the only two cases where a collision in the output occurs.

\begin{proofof}{Theorem~\ref{thm:distinct}}
Let $n > m \geq 2$ be fixed. Let $\bS_m$ be the de Bruijn sequence of order $m$ produced by an $FSR$ with feedback function $h$. Let $f$ be defined based on $h$ as in Equation (\ref{equ-db}). From the state graph $\cG_f$ we let $\cC_f$ and $\cR_f$ be the subgraphs defined earlier. Arithmetic operations on the indices are taken modulo $2^m$ in this proof.
	
Let $\bu_j : j \in \bbra{2^m-1}$ be the \emph{consecutive states} of length $n$ in the directed cycle $\cC_f \subset \cG_f$. Choose one of the vertices, which are in a one-to-one correspondence with the $m$-stage states of $\bS_m$, arbitrarily as $\bu_0$. Recall that $T_j$ is the largest tree in $\cR_f$ whose root is $\bu_{j}$. The vertex sets of the trees $T_j : j \in \bbra{2^m-1}$ partition the set of vertices of $\cG_f$.
	
We choose an arbitrary index $i \in \bbra{2^m-1}$ and let $\bS_n^i$ be the de Bruijn sequence generated by the algorithm on initial state $\bu_i$. Lemma \ref{lemma1} implies that, among the states $\bu_0, \bu_1, \ldots, \bu_{2^m-1}$, the second state that the algorithm visits must be $\bu_{i+1}$. For a contradiction, let $\bu_k$ with $k > i+1$ be the second state visited. Then,  both of its children in $\cG_f$ must have already been visited. This is impossible since one of the children, which is $\bu_{k-1} \neq \bu_i$, has not been visited yet. Thus, by their order of appearance in $\bS_n$, the states are
\[
\bu_i, \bu_{i+1},  \ldots, \bu_{2^m-1}, \bu_0, \ldots, \bu_{i-1}.
\]
	
Combining this fact and Lemma \ref{lemma1} makes it clear that by the time $\bu_j$ is visited, \underline{all} states in $T_{\ell}$ for all $\ell \in \bbra{i+1,j}$ must have been visited. So in the remaining run of the algorithm, each state visited after $\bu_j$ must belong only to $T_{t}$ for \underline{some} $t \in \bbra{j+1,i}$. Similarly, suppose that we run the algorithm with initial state $\bu_j$ for some $j \neq i$ to generate $\bS^j_n$. By the time $\bu_i$ is visited, all states in $T_{\ell}$ for $\ell \in \bbra{j+1,i}$ must have been visited. Each of the remaining states to visit belongs only to $T_{s}$ for some $s \in \bbra{i+1,j}$.
	
We now examine what may allow $\bS_n^i = \bS_n^j$ while $i \neq j$. Without lost of generality, let $i < j$. If $\bS_n^i = \bS_n^j$, then we can partition the set of consecutive states visited by the algorithm on initial state $\bu_i$ into two parts. Part $A$ contains the sequence of states starting from the successor of $\bu_i$ in $\bS^i_n$ up until the state $\bu_j$. This includes all states in $T_{\ell}$ for $\ell \in \bbra{i+1,j}$. Part $B$ hosts the sequence of states starting from the successor of $\bu_j$ in $\bS_n$ until the state $\bu_i$. This includes all states in $T_{\ell}$ for $\ell \in \bbra{j+1,i}$. We say that the two parts are {\it self-closed} since all of the successors of each state in Part $A$, except for $\bu_j$, are also contained in Part $A$. Similarly with Part $B$.
	
Going back to $\cG_f$, we consider each state $\ba$ in Part $A$, except for $\bu_j$, and check if the successor of $\ba$ is a leaf. If yes, then we identify the corresponding tree in $\cR_f$ by its root. Because Part $A$ is self-closed, we claim that all trees corresponding to elements in Part $A$ must also be in Part $A$. To confirm this claim, we use the bijection between $\bu_{\ell}$ and $\bm_{\ell}$ for $\ell \in \bbra{2^m-1}$ to associate
\begin{equation}\label{eq:assoc}
\bu_{\ell}:=u_0,u_1,\ldots,u_{m-1}, \ldots,u_{n-1} \longleftrightarrow \bm_{\ell}:=u_0,u_1,\ldots,u_{m-1}.
\end{equation}
and let $V_{\ell}:=\{\bu_{\ell-1}\} \cup \{\bv \in T_{\ell} : \bv \neq \bu_{\ell}\}$. This allows us to identify $T_{\ell}:=T_{\bu_{\ell}}$ as $T_{\bm_{\ell}}$ and vice versa.
	
The algorithm's assignment rule requires that the respective successors of the conjugate states $\bv:=v_0,v_1,\ldots,v_{n-1}$ and $\overline{\bv}:=v_0+1, v_1, \ldots, v_{n-1}$ must be either one of the two companion states $\bw:=v_1,\ldots,v_{n-1},h(v_{n-m},v_{n-m+1},\ldots,v_{n-1})$ and $\widehat{\bw}:= v_1,\ldots,v_{n-1},h(v_{n-m},v_{n-m+1},\ldots,v_{n-1})+1$. Notice that $\widehat{\bw}$ must be a leaf in the tree whose root has, as its first $m$ bits, 
\[
v_{n-m+1}, \ldots, v_{n-1}, h(v_{n-m},v_{n-m+1},\ldots,v_{n-1})+1.
\]
The latter is the companion state of the child of $(v_{n-m},\ldots,v_{n-1})$ in $\cG_h$.
	
Let a pair of conjugate states $\bv$ and $\overline{\bv}$ whose common last $m$ bits is $\bm_{\ell}$ be given. One of their two possible successors must be a leaf in $T_{\widehat{\bm_{\ell+1}}}$ whose root is the companion state $\widehat{\bm_{\ell+1}}$ of $\bm_{\ell+1}$. We generalize this observation to vertices in $V_{i+1}=\{\bu_i\} \cup \{\bv \in T_{i+1}: \bv \neq \bu_{i+1} \}$. All of the states comes in conjugate pairs whose respective last $m$ bits are $\bm_{i+1}, \bm_{i+2}, \ldots, \bm_{i+n-m}$. Each pair has a state whose successor is a leaf in $\cG_f$. Enumerating the corresponding trees, we obtain
\[
T_{\widehat{\bm_{i+2}}}, T_{\widehat{\bm_{i+3}}}, \ldots, T_{\widehat{\bm_{i+1+n-m}}}.
\]
Hence, going through each conjugate pair in $\displaystyle{\cup_{\ell \in \bbra{i+1,j}} V_{\ell}}$ and identifying the tree that contains a successor which is a leaf gives us the following list of trees:
\[
\bigcup_{\ell \in \bbra{i+1,j}} \{T_{\widehat{\bm_{\ell+1}}}, T_{\widehat{\bm_{\ell+2}}}, \ldots, T_{\widehat{\bm_{\ell+n-m}}}\}.
\]
	
Part $A$ must then have the following property. Because it is self-closed, it contains not only all of the states of the trees $T_{\ell}$ for $\ell \in \bbra{i+1,j}$, but also all states of the trees with respective roots $\displaystyle{\widehat{\bm_{i+2}}, \widehat{\bm_{i+3}},\ldots, \widehat{\bm_{j+n-m}}}$. In total, the states in Part $A$ come from
\[
j+n-m-(i+2)+1=j-i+(n-m-1)
\]
trees. The assumption that $\bS_n^i = \bS_n^j$, however, implies that there are $j-i$ distinct trees that contribute their states to Part $A$. This is impossible if $n > m+1$. Thus, whenever $n > m+1$, distinct initial states $\bu_i$ and $\bu_j$ generate distinct de Bruijn sequences.
	
Now, let $n=m+1$ and $\bu_i$ be the initial state. Suppose that Part $A$ contains states contributed by only one tree and let $\bu_i:=u_0,u_1,\ldots,u_{n-1}$. Then $\bu_{i+1}:=u_1,\ldots,u_{n-1},u_n$ and the only relevant leaf (in $T_{i+1}$) is $\overline{\bu_i}:=u_0+1,u_1,\ldots,u_{n-1}$, which is the successor of $\bu_i$ in $\bS_n^i$. Thus, we have $u_1,u_2,\ldots,u_{n-1},u_n+1=u_0+1,u_1,\ldots,u_{n-1}$, $\ie$,
\[
u_0\neq u_1=u_2=\ldots=u_{n-1}\neq u_n.
\]
There are only two cases that satisfy this constraint.
\begin{enumerate}
\item Part $A$ includes $T_{\0^m1}$, containing $\0^{m+1}$ and $\0^{m}1$, while Part $B$ contains all of the other trees.
\item Part $A$ consists of $T_{\1^m0}$, containing $\1^{m+1}$ and $\1^{m}0$, while Part $B$ contains all of the other trees.
\end{enumerate}
Starting with $\bb=\0^m1$, the last state of the generated de Bruijn sequence must be $\0^{m+1}$. If we begin with $\bb=1\0^m$, then the second and third states of the resulting de Bruijn sequence must be $\0^{m+1}$ and $\0^m1$, respectively. These two de Bruijn sequences are shift equivalent. A similar argument can be made for the initial states $\1^m0$ and $0\1^m$.
	
If Part $A$ fully contains the trees $T_i$ and $T_{i+1}$, then the above analysis confirms that it also contains leaves belonging to the trees $T_{\widehat{\bm_{i+1}}}$ and $T_{\widehat{\bm_{i+2}}}$. Because $\bm_{i+1} \neq \widehat{\bm_{i+1}}$ we have
\[
\bm_i=\widehat{\bm_{i+1}} \mbox{ and } \bm_{i+1} = \widehat{\bm_{i+2}} \implies \bm_i = \bm_{i+2},
\]
which is impossible since $\bS_m$ is de Bruijn. One can proceed inductively to come to the same conclusion for the cases where there are more than $2$ trees that contribute their states to Part $A$. This completes the proof. \qed
\end{proofof}

\begin{table}[h]
	\caption{GPO List for $n = 4$; Part One of Two}
	\label{table:GPO4}
	\setlength{\tabcolsep}{0.1cm}
	\centering
	\scriptsize
	\begin{tabular}{c  c c | c c}
		\hline
		No. & $f$ for Sequence $(0000~1111~0110~0101)$ & $\bb \in $ & $f$ for Sequence $(0000~1011~1101~0011)$ & $\bb \in $\\
		\hline
		$1$ & $x_1 \cdot x_2 \cdot x_3 + x_1 \cdot x_2 + x_1 \cdot x_3 + x_1 + x_2 \cdot x_3 + x_2 + x_3 + 1 $ & $\{1000, 0001\}$
		& $x_1 \cdot x_2 \cdot x_3 + x_1 \cdot x_2 + x_1 \cdot x_3 + x_1 + x_2 \cdot x_3 + x_2 + x_3 + 1$ & $\{0010\}$\\
		
		$2$ & $x_1 \cdot x_2 \cdot x_3 + x_1 \cdot x_2 + x_1 \cdot x_3 + x_1 + x_2 + 1$ & $\{1100\}$
		& $x_1 \cdot x_2 \cdot x_3 + x_1 \cdot x_3 + x_1 + x_2 + x_3 + 1 $ & $\{0100\}$\\
		
		$3$ & $x_1 \cdot x_2 \cdot x_3 + x_1 \cdot x_2 + x_1 + x_2 \cdot x_3 + x_2 + 1$ & $\{0111,1110\}$
		& $x_1 \cdot x_2 + x_1 \cdot x_3 + x_1 + x_2 \cdot x_3 + x_3 + 1 $ & $\{1010\}$\\

		$4$ & $x_1 \cdot x_2 \cdot x_3 + x_1 \cdot x_3 + x_1 + x_2 \cdot x_3 + x_2 + 1$ & $\{1011\}$
		& $x_1 \cdot x_2 \cdot x_3 + x_1 \cdot x_2 + x_1 + x_2 \cdot x_3 + x_3 + 1 $ & $\{1101\}$\\
		
		$5$ & $x_1 \cdot x_2 \cdot x_3 + x_1 \cdot x_3 + x_2 \cdot x_3 + x_2 + x_3 + 1$ & $\{0010\}$
		& $x_1 \cdot x_2 \cdot x_3 + x_1 \cdot x_2 + x_1 \cdot x_3 + x_2 \cdot x_3 + x_2 + 1 $ & $\{0011\}$\\
		
		$6$ & $x_1 \cdot x_2 + x_1 + x_2 \cdot x_3 + x_2 + 1$ & $\{1111\}$
		& $x_1 \cdot x_2 \cdot x_3 + x_1 + x_2 \cdot x_3 + x_3 + 1 $ & $\{1111\}$\\
		
		$7$ & $x_1 \cdot x_2 + x_1 + x_2 + 1$ & $\{0011\}$
		& $x_1 \cdot x_2 + x_1 + x_2 + 1 $ & $\{1000,0001\}$\\
		
		$8$ & $x_1 \cdot x_3 + x_1 + x_2 + 1$ & $\{0110\}$
		& $x_1 \cdot x_3 + x_1 + x_3 + 1 $ & $\{0101\}$\\
		
		$9$ & $x_1 + x_2 \cdot x_3 + x_2 + 1$ & $\{1101\}$
		& $x_1 + x_2 \cdot x_3 + x_3 + 1 $ & $\{0111, 1110\}$\\
		
		$10$ & $x_2 \cdot x_3 + x_2 + x_3 + 1$ & $\{0100\}$
		& $x_1 \cdot x_2 \cdot x_3 + x_1 + x_3 + 1 $ & $\{1011\}$\\
		
		$11$ & $x_1 \cdot x_2 + x_1 \cdot x_3 + x_3 + 1$ & $\{0101\}$
		& $x_1 \cdot x_2 + x_2 + x_3 + 1 $ & $\{1001\}$\\
		
		$12$ & $x_1 \cdot x_2 \cdot x_3 + x_1 \cdot x_2 + x_3 + 1$ & $\{1010\}$
		& $x_1 \cdot x_2 + x_1 \cdot x_3 + x_2 + 1 $ & $\{0110\}$\\
		
		$13$ & $x_2 + 1$ & $\{0011\}$ & $x_1 \cdot x_2 \cdot x_3 + x_1 \cdot x_3 + x_2 + 1 $ & $\{1100\}$\\
		
		$14$ & $0$ & $\{0000\}$ & $x_1 \cdot x_2 \cdot x_3 + x_1 \cdot x_3 + x_2 \cdot x_3 + x_3 $ & $\{0000\}$\\
		\hline
		& $f$ for Sequence $(0000~1010~0111~1011)$ & & $f$ for Sequence $(0000~1111~0101~1001)$ & \\
		\hline
		$1$ & $x_1 \cdot x_2 \cdot x_3 + x_1 \cdot x_2 + x_1 \cdot x_3 + x_1 + x_2 \cdot x_3 + x_2 + x_3 + 1$ & $\{0100\}$ & $x_1 \cdot x_2 \cdot x_3 + x_1 \cdot x_2 + x_1 \cdot x_3 + x_1 + x_2 \cdot x_3 + x_2 + 1$ & $\{1111\}$ \\
		
		$2$ & $x_1 \cdot x_2 + x_1 + x_2 \cdot x_3 + x_2 + x_3 + 1 $ & $\{0010\}$ &
		$x_1 \cdot x_2 + x_1 + x_2 \cdot x_3 + x_2 + x_3 + 1 $ & $\{1000, 0001\}$ \\
		
		$3$ & $x_1 \cdot x_2 \cdot x_3 + x_1 \cdot x_2 + x_1 \cdot x_3 + x_1 + x_2 + 1  $ & $\{1000, 0001\}$ & $x_1 \cdot x_2 + x_1 \cdot x_3 + x_1 + x_2 \cdot x_3 + x_2 + 1 $ & $\{0111, 1110\}$ \\
		
		$4$ & $x_1 \cdot x_2 \cdot x_3 + x_1 \cdot x_2 + x_1 \cdot x_3 + x_2 \cdot x_3 + x_2 + 1 $ & $\{1101\}$ & $x_1 \cdot x_2 \cdot x_3 + x_1 \cdot x_2 + x_1 \cdot x_3 + x_1 + x_2 + 1 $ & $\{0011\}$ \\
		
		$5$ & $x_1 \cdot x_2 \cdot x_3 + x_1 \cdot x_3 + x_2 \cdot x_3 + x_2 + 1 $ & $\{1111\}$ & $x_1 \cdot x_2 \cdot x_3 + x_1 \cdot x_3 + x_1 + x_2 \cdot x_3 + x_2 + 1 $ & $\{1101\}$ \\
		
		$6$ & $x_1 \cdot x_3 + x_1 + x_3 + 1 $ & $\{1010\}$ & $x_1 \cdot x_2 \cdot x_3 + x_1 \cdot x_2 + x_1 \cdot x_3 + x_1 + x_2 \cdot x_3 + 1 $ & $\{0110\}$ \\
		
		$7$ & $x_1 \cdot x_2 \cdot x_3 + x_1 + x_3 + 1 $ & $\{0101\}$ &
		$x_1 \cdot x_2 \cdot x_3 + x_1 \cdot x_3 + x_2 \cdot x_3 + x_2 + x_3 + 1 $ & $\{0100\}$ \\
		
		$8$ & $x_2 \cdot x_3 + x_2 + x_3 + 1 $ & $\{1001\}$ & $x_1 + x_2 \cdot x_3 + x_2 + 1 $ & $\{1010\}$ \\
		
		$9$ & $x_1 \cdot x_2 + x_2 \cdot x_3 + x_2 + 1 $ & $\{1011\}$ & $x_1 \cdot x_2 + x_1 \cdot x_3 + x_1 + 1  $ & $\{1011\}$ \\
		
		$10$ & $x_1 \cdot x_2 \cdot x_3 + x_1 \cdot x_2 + x_2 + 1 $ & $\{0110\}$ & $x_1 \cdot x_2 \cdot x_3 + x_1 \cdot x_2 + x_1 + 1  $ & $\{0101\}$ \\
		
		$11$ & $x_1 \cdot x_3 + x_2 \cdot x_3 + x_2 + 1 $ & $\{0111, 1110\}$ & $x_1 \cdot x_3 + x_1 + x_2 \cdot x_3 + 1 $ & $\{1100\}$ \\
		
		$12$ & $x_1 \cdot x_2 \cdot x_3 + x_1 \cdot x_3 + x_2 + 1 $ & $\{0011\}$ & $x_1 \cdot x_2 + x_1 \cdot x_3 + x_3 + 1 $ & $\{0010\}$ \\
		
		$13$ & $x_2 + 1  $ & $\{1100\}$ & $x_1 \cdot x_2 \cdot x_3 + x_1 \cdot x_2 + x_2 \cdot x_3 + 1 $ & $\{1001\}$ \\
		
		$14$ & $x_2 \cdot x_3 + x_3  $ & $\{0000\}$ & $x_1 \cdot x_2 \cdot x_3 + x_1 \cdot x_3 $ & $\{0000\}$ \\
		
		\hline
		& $f$ for Sequence $(0000~1101~1110~0101)$ & & $f$ for Sequence $(0000~1011~0100~1111)$ &\\
		\hline
		$1$ & $x_1 \cdot x_2 + x_1 \cdot x_3 + x_1 + x_2 + x_3 + 1$ & $\{1000, 0001\}$ &
		$x_1 \cdot x_2 + x_1 \cdot x_3 + x_1 + x_2 + x_3 + 1 $ & $\{0010\}$ \\
		
		$2$ & $x_1 \cdot x_2 + x_1 \cdot x_3 + x_1 + x_2 \cdot x_3 + x_2 + 1 $ & $\{1100\}$ & $x_1 \cdot x_3 + x_1 + x_2 \cdot x_3 + x_2 + x_3 + 1 $ & $\{0100\}$ \\
		
		$3$ & $x_1 \cdot x_2 \cdot x_3 + x_1 \cdot x_2 + x_1 + x_2 \cdot x_3 + x_2 + 1 $ & $\{0011\}$ & $x_1 \cdot x_2 \cdot x_3 + x_1 \cdot x_2 + x_1 + x_2 \cdot x_3 + x_2 + 1 $ & $\{1000, 0001\}$ \\
		
		$4$ & $x_1 \cdot x_2 \cdot x_3 + x_1 \cdot x_3 + x_1 + x_2 \cdot x_3 + x_2 + 1 $ & $\{0111, 1110\}$ & $x_1 \cdot x_2 \cdot x_3 + x_1 \cdot x_2 + x_1 \cdot x_3 + x_1 + x_3 + 1 $ & $\{1010\}$ \\
		
		$5$ & $x_1 \cdot x_2 \cdot x_3 + x_1 \cdot x_2 + x_1 \cdot x_3 + x_2 \cdot x_3 + x_3 + 1 $ & $\{0101\}$ & $x_1 \cdot x_2 \cdot x_3 + x_1 \cdot x_3 + x_1 + x_2 \cdot x_3 + x_3 + 1 $ & $\{0101\}$ \\
		
		$6$ & $x_1 \cdot x_3 + x_1 + x_2 \cdot x_3 + x_2 + 1 $ & $\{1111\}$ & $x_1 \cdot x_2 \cdot x_3 + x_1 \cdot x_2 + x_2 \cdot x_3 + x_2 + x_3 + 1  $ & $\{1001\}$ \\
		
		$7$ & $x_1 \cdot x_2 + x_1 + x_2 + 1  $ & $\{0110\}$ & $x_1 \cdot x_2 \cdot x_3 + x_1 \cdot x_2 + x_1 \cdot x_3 + x_2 \cdot x_3 + x_2 + 1 $ & $\{0111, 1110\}$ \\
		
		$8$ & $x_1 \cdot x_3 + x_1 + x_2 + 1  $ & $\{1011\}$ & $x_1 \cdot x_2 + x_1 \cdot x_3 + x_2 \cdot x_3 + x_2 + 1  $ & $\{1111\}$ \\
		
		$9$ & $x_1 \cdot x_2 \cdot x_3 + x_1 + x_2 + 1 $ & $\{1101\}$ & $x_1 \cdot x_2 + x_1 + x_3 + 1 $ & $\{1101\}$ \\
		
		$10$ & $x_1 \cdot x_3 + x_2 + x_3 + 1  $ & $\{0010\}$ & $x_1 + x_2 \cdot x_3 + x_3 + 1  $ & $\{1011\}$ \\
		
		$11$ & $x_1 \cdot x_2 \cdot x_3 + x_2 + x_3 + 1  $ & $\{0100\}$ & $x_1 \cdot x_2 \cdot x_3 + x_1 + x_3 + 1 $ & $\{0110\}$ \\
		
		$12$ & $x_1 \cdot x_2 \cdot x_3 + x_2 \cdot x_3 + x_2 + 1  $ & $\{1001\}$ & $x_1 \cdot x_2 + x_1 \cdot x_3 + x_2 + 1 $ & $\{0011\}$ \\
		
		$13$ & $x_1 \cdot x_2 + x_2 \cdot x_3 + x_3 + 1  $ & $\{1010\}$ & $x_1 \cdot x_3 + x_2 \cdot x_3 + x_2 + 1 $ & $\{1100\}$ \\
		
		$14$ & $x_1 \cdot x_2 \cdot x_3 + x_2 \cdot x_3 $ & $\{0000\}$ & $x_1 \cdot x_3 + x_3  $ & $\{0000\}$ \\
		
		\hline
		& $f$ for Sequence $(0000~1010~0110~1111)$ & & $f$ for Sequence $(0000~1001~1110~1011)$ &\\
		\hline
		$1$ & $x_1 \cdot x_2 + x_1 \cdot x_3 + x_1 + x_2 + x_3 + 1 $ & $\{0100\}$
		& $x_1 \cdot x_2 + x_1 + x_2 \cdot x_3 + x_2 + x_3 + 1$ & $\{0100\}$ \\
		
		$2$ & $x_1 \cdot x_2 \cdot x_3 + x_1 \cdot x_2 + x_1 + x_2 + x_3 + 1 $ & $\{0010\}$ & $x_1 \cdot x_2 \cdot x_3 + x_1 \cdot x_3 + x_2 \cdot x_3 + x_2 + x_3 + 1 $ & $\{1001\}$ \\
		
		$3$ & $x_1 \cdot x_2 + x_1 \cdot x_3 + x_1 + x_2 \cdot x_3 + x_2 + 1  $ & $\{1000, 0001\}$ & $x_1 \cdot x_2 \cdot x_3 + x_1 \cdot x_2 + x_1 \cdot x_3 + x_2 \cdot x_3 + x_2 + 1 $ & $\{1010\}$ \\
		
		$4$ & $x_1 \cdot x_2 \cdot x_3 + x_1 \cdot x_3 + x_1 + x_2 \cdot x_3 + x_3 + 1 $ & $\{1010 \}$ & $x_1 \cdot x_2 \cdot x_3 + x_1 + x_3 + 1 $ & $\{0010\}$ \\
		
		$5$ & $x_1 \cdot x_2 \cdot x_3 + x_1 \cdot x_2 + x_2 \cdot x_3 + x_2 + 1 $ & $\{1111\}$ & $x_1 \cdot x_3 + x_1 + x_2 \cdot x_3 + 1 $ & $\{1000, 0001\}$ \\
		
		$6$ & $x_1 + x_2 \cdot x_3 + x_3 + 1 $ & $\{0101\}$ & $x_1 \cdot x_2 + x_2 \cdot x_3 + x_2 + 1 $ & $\{1101\}$ \\
		
		$7$ & $x_1 \cdot x_2 \cdot x_3 + x_2 + x_3 + 1 $ & $\{1001\}$ & $x_1 \cdot x_2 \cdot x_3 + x_2 \cdot x_3 + x_2 + 1 $ & $\{0111, 1110\}$ \\
		
		$8$ & $x_1 \cdot x_2 + x_1 \cdot x_3 + x_2 + 1 $ & $\{1101\}$ & $x_1 \cdot x_2 \cdot x_3 + x_1 \cdot x_2 + x_2 \cdot x_3 + 1  $ & $\{1100\}$ \\
		
		$9$ & $x_1 \cdot x_2 + x_2 \cdot x_3 + x_2 + 1 $ & $\{0111, 1110\}$ & $x_1 \cdot x_2 \cdot x_3 + x_1 \cdot x_2 + x_2 + x_3  $ & $\{0000\}$ \\
		
		$10$ & $x_1 \cdot x_2 \cdot x_3 + x_1 \cdot x_2 + x_2 + 1 $ & $\{1011\}$ & $x_2 \cdot x_3 + x_2 + 1  $ & $\{1111\}$ \\
		
		$11$ & $x_1 \cdot x_3 + x_2 \cdot x_3 + x_2 + 1 $ & $\{0011\}$ & $x_2 + 1 $ & $\{0011\}$ \\
		
		$12$ & $x_1 \cdot x_2 \cdot x_3 + x_1 \cdot x_3 + x_2 + 1 $ & $\{0110\}$ & $x_1 \cdot x_3 + 1  $ & $\{0101\}$ \\
		
		$13$ & $x_1 \cdot x_2 \cdot x_3 + x_2 \cdot x_3 + x_2 + 1 $ & $\{1100\}$ & $x_2 \cdot x_3 + 1  $ & $\{0110\}$ \\
		
		$14$ & $x_1 \cdot x_2 \cdot x_3 + x_3 $ & $\{0000\}$ & $x_1 \cdot x_2 \cdot x_3 + 1 $ & $\{1011\}$ \\
		
		\hline
	\end{tabular}
\end{table}

\begin{table}[h]
	\caption{GPO List for $n = 4$; Part Two of Two}
	\label{table:GPO4part2}
	\centering
	\scriptsize
	\begin{tabular}{c  c c | c c}
		\hline
		No. & $f$ for Sequence $(0000~1101~0111~1001)$  & $\bb \in $ & $f$ for Sequence $(0000~1001~1010~1111)$ & $\bb \in $\\
		\hline
		$1$ & $x_1 \cdot x_2 \cdot x_3 + x_1 \cdot x_2 + x_1 + x_2 + x_3 + 1$ & $\{1000, 0001\}$ & $x_1 \cdot x_2 \cdot x_3 + x_1 \cdot x_2 + x_1 + x_2 + x_3 + 1$ & $\{0100\}$ \\
		
		$2$ & $x_1 \cdot x_2 + x_1 \cdot x_3 + x_1 + x_2 \cdot x_3 + x_2 + 1 $ & $\{0011\}$ & $x_1 + x_2 \cdot x_3 + x_3 + 1 $ & $\{0010\}$ \\
		
		$3$ & $x_1 \cdot x_2 \cdot x_3 + x_1 \cdot x_2 + x_1 \cdot x_3 + x_1 + x_2 + 1 $ & $\{0110\}$ & $x_1 \cdot x_2 \cdot x_3 + x_1 \cdot x_3 + x_1 + 1 $ & $\{1000, 0001\}$ \\
		
		$4$ & $x_1 \cdot x_2 \cdot x_3 + x_1 \cdot x_2 + x_1 \cdot x_3 + x_1 + x_2 \cdot x_3 + 1 $ & $\{1011\}$ & $x_1 \cdot x_3 + x_2 + x_3 + 1  $ & $\{1001\}$ \\
		
		$5$ & $x_1 \cdot x_2 \cdot x_3 + x_1 \cdot x_2 + x_1 \cdot x_3 + x_2 \cdot x_3 + x_3 + 1 $ & $\{0010\}$ & $x_1 \cdot x_2 + x_1 \cdot x_3 + x_2 + 1  $ & $\{1010\}$ \\
		
		$6$ & $x_1 \cdot x_2 \cdot x_3 + x_1 \cdot x_2 + x_1 \cdot x_3 + x_1 + 1 $ & $\{1111\}$ & $x_1 \cdot x_2 \cdot x_3 + x_1 \cdot x_2 + x_2 + 1 $ & $\{1101\}$ \\
		
		$7$ & $x_1 \cdot x_3 + x_1 + x_2 + 1 $ & $\{1101\}$ & $x_1 \cdot x_2 \cdot x_3 + x_2 \cdot x_3 + x_2 + 1  $ & $\{0011\}$ \\
		
		$8$ & $x_1 \cdot x_2 \cdot x_3 + x_1 + x_2 + 1 $ & $\{1010\}$ & $x_1 \cdot x_2 \cdot x_3 + x_1 \cdot x_3 + x_2 \cdot x_3 + 1 $ & $\{0101, 0000\}$ \\
		
		$9$ & $x_1 \cdot x_2 + x_1 \cdot x_3 + x_1 + 1 $ & $\{0111, 1110\}$ & $x_2 + 1 $ & $\{0110, 1100\}$ \\
		
		$10$ & $x_1 \cdot x_2 + x_1 + x_2 \cdot x_3 + 1 $ & $\{0101\}$ & $x_2 \cdot x_3 + 1  $ & $\{1011\}$ \\
		
		$11$ & $x_1 \cdot x_2 \cdot x_3 + x_1 \cdot x_3 + x_1 + 1 $ & $\{1100\}$ & $x_1 \cdot x_2 \cdot x_3 + 1  $ & $\{0111, 1110\}$ \\
		
		$12$ & $x_1 \cdot x_3 + x_2 + x_3 + 1  $ & $\{0100\}$ & $1 $ & $\{1111\}$ \\
		
		$13$ & $x_1 \cdot x_2 + 1  $ & $\{1001\}$ & &  \\
		
		$14$ & $x_1 \cdot x_3 + x_2 \cdot x_3 $ & $\{0000\}$ & & \\
		\hline
		& $f$ for Sequence $(0000~1111~0010~1101)$ & & $f$ for Sequence $(0000~1011~1100~1101)$ & \\
		\hline
		$1$ & $x_1 \cdot x_3 + x_1 + x_2 \cdot x_3 + x_2 + x_3 + 1 $ & $\{1000, 0001\}$ &
		$x_1 \cdot x_3 + x_1 + x_2 \cdot x_3 + x_2 + x_3 + 1 $ & $\{0010\}$ \\
		
		$2$ & $x_1 \cdot x_2 \cdot x_3 + x_1 \cdot x_2 + x_1 + x_2 \cdot x_3 + x_2 + 1 $ & $\{1100\}$ &
		$x_1 \cdot x_2 \cdot x_3 + x_1 \cdot x_2 + x_1 \cdot x_3 + x_1 + x_3 + 1 $ & $\{0101\}$ \\
		
		$3$ & $x_1 \cdot x_2 \cdot x_3 + x_1 \cdot x_2 + x_2 \cdot x_3 + x_2 + x_3 + 1 $ & $\{0100\}$ &
		$x_1 \cdot x_2 \cdot x_3 + x_1 \cdot x_2 + x_1 + x_2 \cdot x_3 + x_3 + 1 $ & $\{0111, 1110\}$ \\
		
		$4$ & $x_1 \cdot x_2 \cdot x_3 + x_1 \cdot x_2 + x_1 \cdot x_3 + x_2 \cdot x_3 + x_3 + 1$ & $\{ 1011\}$ &
		$x_1 \cdot x_2 \cdot x_3 + x_1 \cdot x_2 + x_1 \cdot x_3 + x_2 \cdot x_3 + x_3 + 1 $ & $\{1001\}$ \\
		
		$5$ & $x_1 \cdot x_2 \cdot x_3 + x_1 + x_2 \cdot x_3 + x_2 + 1 $ & $\{1111\}$ &
		$x_1 \cdot x_2 + x_1 + x_2 \cdot x_3 + x_3 + 1  $ & $\{111\}$ \\
		
		$6$ & $x_1 + x_2 \cdot x_3 + x_2 + 1 $ & $\{0111, 1110\}$ &
		$x_1 \cdot x_2 \cdot x_3 + x_1 + x_2 + 1 $ & $\{1000, 0001\}$ \\
		
		$7$ & $x_1 \cdot x_2 \cdot x_3 + x_1 + x_2 + 1  $ & $\{0011\}$ &
		$x_1 \cdot x_2 + x_1 + x_3 + 1  $ & $\{1011\}$ \\
		
		$8$ & $x_1 \cdot x_2 \cdot x_3 + x_2 + x_3 + 1 $ & $\{0010\}$ &
		$x_1 + x_2 \cdot x_3 + x_3 + 1 $ & $\{1100\}$ \\
		
		$9$ & $x_1 \cdot x_3 + x_2 \cdot x_3 + x_2 + 1  $ & $\{1001\}$ &
		$x_1 \cdot x_2 + x_1 \cdot x_3 + x_2 + 1 $ & $\{0100\}$ \\
		
		$10$ & $x_1 \cdot x_2 + x_1 \cdot x_3 + x_3 + 1 $ & $\{0110\}$ &
		$x_1 \cdot x_2 \cdot x_3 + x_1 \cdot x_2 + x_2 \cdot x_3 + 1 $ & $\{0110\}$ \\
		
		$11$ & $x_1 \cdot x_2 + x_2 \cdot x_3 + x_3 + 1 $ & $\{0101\}$ &
		$x_1 \cdot x_2 \cdot x_3 + x_1 \cdot x_3 + x_2 \cdot x_3 + 1 $ & $\{1010\}$ \\
		
		$12$ & $x_1 \cdot x_2 \cdot x_3 + x_1 \cdot x_3 + x_3 + 1 $ & $\{1101\}$ &
		$x_1 \cdot x_2 + x_1 \cdot x_3 + x_2 \cdot x_3 + x_3 $ & $\{0000\}$ \\
		
		$13$ & $x_3 + 1 $ & $\{1010\}$ &
		$x_1 \cdot x_2 + 1   $ & $\{0011\}$ \\
		
		$14$ & $x_1 \cdot x_2 \cdot x_3 + x_1 \cdot x_2 $ & $\{0000\}$ &
		$x_2 \cdot x_3 + 1  $ & $\{1101\}$ \\
		\hline
		& $f$ for Sequence $(0000~1100~1011~1101)$ & & $f$ for Sequence $(0000~1011~0011~1101)$ & \\
		\hline
		$1$ & $x_1 \cdot x_2 \cdot x_3 + x_1 \cdot x_3 + x_1 + x_2 + x_3 + 1$ & $\{1000, 0001\}$ &
		$x_1 \cdot x_2 \cdot x_3 + x_1 \cdot x_3 + x_1 + x_2 + x_3 + 1 $ & $\{0010\}$ \\
		
		$2$ & $x_1 \cdot x_2 \cdot x_3 + x_1 \cdot x_2 + x_1 \cdot x_3 + x_2 \cdot x_3 + x_3 + 1 $ & $\{0111, 1110\}$ &
		$x_1 \cdot x_2 + x_1 \cdot x_3 + x_1 + x_2 \cdot x_3 + x_3 + 1 $ & $\{0101\}$ \\
		
		$3$ & $x_1 \cdot x_2 + x_1 \cdot x_3 + x_2 \cdot x_3 + x_3 + 1 $ & $\{1111\}$ &
		$x_1 \cdot x_2 \cdot x_3 + x_1 \cdot x_2 + x_1 + x_2 \cdot x_3 + x_3 + 1 $ & $\{1011\}$ \\
		
		$4$ & $x_1 \cdot x_2 + x_1 + x_2 + 1 $ & $\{1100\}$ &
		$x_1 \cdot x_2 \cdot x_3 + x_1 \cdot x_2 + x_1 \cdot x_3 + x_2 \cdot x_3 + x_2 + 1 $ & $\{0100\}$ \\
		
		$5$ & $x_1 + x_2 \cdot x_3 + x_2 + 1 $ & $\{0011\}$ &
		$x_1 + x_2 \cdot x_3 + x_2 + 1   $ & $\{1000, 0001\}$ \\
		
		$6$ & $x_1 \cdot x_2 \cdot x_3 + x_1 + x_2 + 1 $ & $\{0110\}$ &
		$x_1 \cdot x_2 + x_1 + x_3 + 1  $ & $\{0110\}$ \\
		
		$7$ & $x_1 \cdot x_2 + x_2 + x_3 + 1 $ & $\{0100\}$ &
		$x_1 \cdot x_2 \cdot x_3 + x_1 + x_3 + 1   $ & $\{1100\}$ \\
		
		$8$ & $x_2 \cdot x_3 + x_2 + x_3 + 1 $ & $\{0010\}$ &
		$x_1 \cdot x_2 + x_1 \cdot x_3 + x_3 + 1 $ & $\{1001\}$ \\
		
		$9$ & $x_1 \cdot x_2 \cdot x_3 + x_1 \cdot x_3 + x_2 + 1 $ & $\{1001\}$ &
		$x_1 \cdot x_2 \cdot x_3 + x_1 \cdot x_2 + x_2 \cdot x_3 + 1 $ & $\{0011\}$ \\
		
		$10$ & $x_1 \cdot x_2 + x_1 \cdot x_3 + x_3 + 1 $ & $\{1011\}$ &
		$x_1 \cdot x_2 \cdot x_3 + x_1 \cdot x_2 + x_1 \cdot x_3 + x_3 $ & $\{0000\}$ \\
		
		$11$ & $x_1 \cdot x_2 \cdot x_3 + x_1 \cdot x_2 + x_3 + 1 $ & $\{0101\}$ &
		$x_1 \cdot x_2 \cdot x_3 + x_1 \cdot x_2 + 1 $ & $\{1111\}$ \\
		
		$12$ & $x_1 \cdot x_3 + x_2 \cdot x_3 + x_3 + 1  $ & $\{1101\}$ &
		$x_1 \cdot x_2 + 1  $ & $\{0111, 1110\}$ \\
		
		$13$ & $x_1 \cdot x_2 \cdot x_3 + x_2 \cdot x_3 + x_3 + 1$ & $\{1010\}$ &
		$x_1 \cdot x_3 + 1  $ & $\{1010\}$ \\
		
		$14$ & $x_1 \cdot x_2 + x_2 \cdot x_3 $ & $\{0000\}$ &
		$x_1 \cdot x_2 \cdot x_3 + 1 $ & $\{1101\}$ \\
		\hline
		& $f$ for Sequence $(0000~1101~0010~1111)$ & & $f$ for Sequence $(0000~1111~0100~1011)$ &\\
		\hline
		$1$ & $x_1 \cdot x_2 \cdot x_3 + x_1 \cdot x_2 + x_1 \cdot x_3 + x_1 + x_2 \cdot x_3 + 1 $ & $\{1101\}$ &
		$x_1 \cdot x_2 \cdot x_3 + x_1 \cdot x_2 + x_1 \cdot x_3 + x_2 \cdot x_3 + x_2 + 1$ & $\{1001\}$ \\
		
		$2$ & $x_1 \cdot x_2 \cdot x_3 + x_1 \cdot x_2 + x_2 \cdot x_3 + x_2 + x_3 + 1$ & $\{0010\}$ &
		$x_1 + x_2 \cdot x_3 + x_2 + 1  $ & $\{0100\}$ \\
		
		$3$ & $x_1 \cdot x_2 \cdot x_3 + x_1 \cdot x_2 + x_1 \cdot x_3 + x_2 \cdot x_3 + x_3 + 1 $ & $\{1100\}$ &
		$x_1 \cdot x_2 \cdot x_3 + x_1 + x_3 + 1 $ & $\{1000, 0001\}$ \\
		
		$4$ & $x_1 \cdot x_2 \cdot x_3 + x_1 \cdot x_3 + x_2 \cdot x_3 + x_3 + 1$ & $\{1111\}$ &
		$x_1 \cdot x_2 + x_1 \cdot x_3 + x_1 + 1 $ & $\{1101\}$ \\
		
		$5$ & $x_1 \cdot x_2 \cdot x_3 + x_1 + x_2 + 1 $ & $\{0100\}$ &
		$x_1 \cdot x_2 \cdot x_3 + x_1 \cdot x_2 + x_1 + 1 $ & $\{1010\}$ \\
		
		$6$ & $x_1 + x_2 \cdot x_3 + x_3 + 1  $ & $\{1000, 0001\}$ &
		$x_1 \cdot x_3 + x_1 + x_2 \cdot x_3 + 1  $ & $\{0011\}$ \\
		
		$7$ & $x_1 \cdot x_2 + x_1 + x_2 \cdot x_3 + 1  $ & $\{1010\}$ &
		$x_1 \cdot x_2 \cdot x_3 + x_1 \cdot x_3 + x_1 + 1$ & $\{0111, 1110\}$ \\
		
		$8$ & $x_1 \cdot x_3 + x_1 + x_2 \cdot x_3 + 1 $ & $\{0110\}$ &
		$x_1 \cdot x_2 + x_2 + x_3 + 1$ & $\{0010\}$ \\
		
		$9$ & $x_1 \cdot x_2 \cdot x_3 + x_1 \cdot x_3 + x_1 + 1 $ & $\{0011\}$ &
		$x_1 \cdot x_2 + x_1 \cdot x_3 + x_3 + 1$ & $\{1100\}$ \\
		
		$10$ & $x_1 \cdot x_2 + x_1 \cdot x_3 + x_2 + 1  $ & $\{1001\}$ &
		$x_1 \cdot x_3 + x_2 \cdot x_3 + x_3 + 1$ & $\{1011\}$ \\
		
		$11$ & $x_1 \cdot x_3 + x_2 \cdot x_3 + x_3 + 1  $ & $\{0111, 1110\}$ &
		$x_1 \cdot x_2 \cdot x_3 + x_1 \cdot x_3 + x_3 + 1$ & $\{0110\}$ \\
		
		$12$ & $x_1 \cdot x_2 \cdot x_3 + x_1 \cdot x_3 + x_3 + 1 $ & $\{1011\}$ &
		$x_1 \cdot x_2 \cdot x_3 + x_2 \cdot x_3 + x_3 + 1$ & $\{0101\}$ \\
		
		$13$ & $x_1 \cdot x_2 \cdot x_3 + x_1 \cdot x_2 + x_1 \cdot x_3 + x_2 $ & $\{0000\}$ &
		$x_1 \cdot x_2 + x_1 \cdot x_3 + x_2 \cdot x_3 + x_2$ & $\{0000\}$ \\
		
		$14$ & $x_3 + 1 $ & $\{0101\}$ &
		$x_1 \cdot x_3 + x_1 + 1 $ & $\{1111\}$ \\
		\hline
	\end{tabular}
\end{table}
\end{document}